\documentclass[10pt,a4paper]{article}
\usepackage{lmodern}
\usepackage[T1]{fontenc}
\usepackage{booktabs}
\usepackage{colortbl}
\usepackage{xcolor}
\usepackage{adjustbox}
\usepackage{array}

\definecolor{imposs}{RGB}{200,40,40}
\definecolor{newcell}{RGB}{0,90,180}

\usepackage{authblk}
\usepackage{geometry}
\usepackage{amsmath,amssymb,amsthm}
\usepackage{graphicx}

\usepackage{hyperref}

\usepackage{latexsym}
\usepackage{algorithm}
\usepackage{algpseudocode}
\usepackage{amssymb,amsthm}
\usepackage{amsfonts}
\usepackage{amsmath}
\usepackage{multirow,eepic}
\usepackage{xspace}
\usepackage{tikz}
\usetikzlibrary{positioning}
\usetikzlibrary{automata}
\usetikzlibrary{arrows}
\usetikzlibrary{calc}
\usetikzlibrary{arrows.meta}
\makeatletter
\def\senbun#1(#2)#3({\@senbun(#2)(}
\def\@senbun(#1,#2)(#3,#4){%
   \@tempdima#1\p@ \advance\@tempdima#3\p@
   \divide\@tempdima\tw@
   \@tempdimb#2\p@ \advance\@tempdimb#4\p@
   \divide\@tempdimb\tw@
   \edef\@senbun@temp{\noexpand\qbezier(#1,#2)%
      (\strip@pt\@tempdima,\strip@pt\@tempdimb)(#3,#4)}%
  \@senbun@temp}
\makeatother
\newcommand{\ASY}{\textnormal{\textsc{Asynch}}\xspace}
\newcommand{\FSY}{\textnormal{\textsc{Fsynch}}\xspace}
\newcommand{\SSY}{\textnormal{\textsc{Ssynch}}\xspace}
\newcommand{\RSY}{\textnormal{\textsc{Rsynch}}\xspace}
\newcommand{\RR}{\textnormal{\textsc{RR}}\xspace}

\newcommand{\LU}{\ensuremath{\mathcal{LUMI}}\xspace}
\newcommand{\FS}{\ensuremath{\mathcal{FST\!A}}\xspace}
\newcommand{\FC}{\ensuremath{\mathcal{FCOM}}\xspace}
\newcommand{\OB}{\ensuremath{\mathcal{OBLOT}}\xspace}
\newcommand{\Look}{\ensuremath{\mathit{Look}}\xspace}
\newcommand{\Compute}{\ensuremath{\mathit{Compute}}\xspace}
\newcommand{\Move}{\ensuremath{\mathit{Move}}\xspace}
\newcommand{\LCM}{\ensuremath{\mathit{LCM}}\xspace}

\newcommand{\N}{{\rm I\kern-.22em N}} 
\newcommand{\Z}{{\sf Z\kern-.42em Z}} 
\newcommand{\R}{{\rm I\kern-.22em R}}

\newcommand{\ML}{\ensuremath{\mathcal L}\xspace}

\newcommand{\QS}{quasi-self-stabilizing}
\newcommand{\SSS}{self-stabilizing}

\newcommand{\obs}{\mathit obs}
\newcommand{\C}{\mathcal{C}}
\newcounter{Codeline}
\newcommand{\Newcodeline}{\setcounter{Codeline}{1}}

\newtheorem{theoremduplicate}{Theorem}
\newtheorem{lemmaduplicate}{Lemma}
\usepackage{graphicx}
\usepackage{tikz}
\usetikzlibrary{automata,positioning,arrows.meta,calc}

\newcommand{\DE}{\color{black}}

\newtheorem{lemma}{Lemma}
\newtheorem{theorem}{Theorem}
\newtheorem{corollary}{Corollary}
\newtheorem{proposition}{Proposition}
\newtheorem{remark}{Remark}
\theoremstyle{definition}

\newtheorem{definition}{Definition}

\title{Rendezvous under Variable Disorientation:\\
The Algorithmic Power of Fixed Unit Distance}
\date{}

\author[1]{Toshimitsu Masuzawa}
\author[2]{Yuichi Sudo}
\author[2]{Koichi Wada}

\affil[1]{Notre Dame Seishin University, Okayama, Japan}
\affil[2]{Hosei University, Tokyo, Japan}

\begin{document}
 
\maketitle

\begin{abstract}
We study deterministic Rendezvous of two robots with lights
under weak geometric and synchronization assumptions.
Our focus is on Variable Disorientation (VD), where each robot
may arbitrarily rotate, reflect, and rescale its local coordinate
system at every \Look, and on VD+Fixed-Unit-Distance (FUD),
where each robot's local unit distance remains fixed.
We consider the one-sided visibility models \(\FS\) and \(\FC\),
and schedulers ranging from the energy-restricted
\RSY\ and R-\RSY\ to fully asynchronous executions.

Under VD, we prove color-independent impossibility results for
\(\FS\) under \SSY\ and for \(\FC\) under \ASY, and
establish sharp two-color boundaries under \RSY\ and
R-\RSY.
Under VD+FUD, we introduce \emph{periodic distance classes},
which turn multiplicative changes of the physical distance into
predictable cyclic phase shifts.
This technique yields a two-color self-stabilizing
non-\ML-Rendezvous algorithm for both \(\FS\) and
\(\FC\) under \RSY\ and R-\RSY, even with Non-Rigid
movement.

Under Rigid movement, periodic distance classes further yield
an optimal three-color non-quasi-self-stabilizing algorithm for
\(\FS\) under \SSY, a four-color \(\FS\) algorithm under
$M$-atomic \ASY, and a four-color \(\FC\) algorithm under
CM-atomic \ASY.
By guarding the latter phase mechanism with two additional
colors and replacing its initial swap by passive initialization,
we obtain a six-color \(\FC\) algorithm under full \ASY,
improving the previous twelve-color upper bound by a factor of
two.

These results show that a fixed local metric can serve as
persistent cyclic memory, and clarify how metric stability,
one-sided visibility, and scheduler atomicity jointly determine
Rendezvous solvability and color complexity.
\end{abstract}

\section{Introduction}

\noindent\textbf{Robot Models.}\enspace
The computational issues of autonomous mobile robots
have been extensively studied in distributed computing~\cite{AP,BDT,CFPS,DKLMPW,IBTW,KLOT,SDY}
since the seminal work of Suzuki and Yamashita~\cite{SY}.
In the standard setting, a robot is modeled as a point in the plane,
and robots are anonymous, uniform, and deterministic.
They operate in \Look-\Compute-\Move (\LCM) cycles:
each activated robot observes the current configuration,
computes a destination and a new internal state,
and then moves accordingly.

To overcome the limitations of oblivious robots,
persistent lights have been introduced,
leading to the models $\OB$~\cite{SY}, $\FS$ and $\FC$~\cite{FSVY}, and $\LU$~\cite{DFPSY}.
In this paper we focus on the two one-sided visibility models:
$\FS$, where a robot observes only its own light,
and $\FC$, where it observes only the other robot’s light.
Although these models differ in the direction of visibility,
many of our arguments depend only on the fact that
exactly one color is visible at each \Look.
This motivates our unified use of the one-sided visibility model.

\noindent\textbf{Geometric Assumptions and Schedulers.}\enspace
We study weak geometric settings.
In particular, under \emph{Variable Disorientation} (VD),
the adversary may arbitrarily rotate, reflect,
and rescale each robot’s local coordinate system at every \Look phase.
We also consider its refinement
\emph{VD+Fixed-Unit-Distance} (FUD),
where the local unit distance remains fixed.
These assumptions capture situations in which
robots do not share a stable geometric reference frame.

We consider both the energy-restricted schedulers \RSY\ and
R-\RSY~\cite{BuFKPSW22} and increasingly asynchronous settings:
\SSY, \(CM\)-atomic and \(M\)-atomic \ASY, and full \ASY~\cite{FSSW23}.
This allows us to study both phase-structured activation and
the gradual loss of atomicity between \Compute and \Move.


\noindent\textbf{Rendezvous and \ML-Algorithms.}\enspace
Rendezvous is the two-robot version of Gathering:
two initially separated robots must meet at an a priori unknown
location in finite time.
It has been extensively studied under various geometric and
scheduling assumptions~\cite{DFPSY,DePP20,FSW19,IKIW,ISKIDWY,LMA,OWD,SDY}.
A central class of Rendezvous algorithms is formed by \ML-algorithms.
In such algorithms, each robot always computes a destination
on the line connecting the two robots, and both the movement
parameter and the light update are determined solely by the
observed color, independently of the distance between the
robots.
Under the one-sided visibility model considered in this paper,
this means that the robot reacts only to the single observed color
at each \Look.
Because \ML-algorithms abstract away from metric information,
they provide a natural reference class
for studying Rendezvous under weak geometric assumptions.


\noindent\textbf{Related Work and Landscape.}\enspace
Previous work has established several sharp possibility
and impossibility results for \ML-algorithms
under geometrically stable settings.
In particular, a number of classical lower and upper bounds
were obtained under fixed disorientation
and fixed unit distance~\cite{FSVY,Vig}.
The present paper revisits this landscape
from the viewpoint of VD and VD+FUD.
Under VD, some of the known impossibility results
for \ML-algorithms remain valid,
and in the scheduler/model combinations considered here
they can even be strengthened beyond the \ML\ setting.
Under VD+FUD, by contrast,
the recovered metric stability enables new positive results.
In this sense, our results provide
a geometric reclassification of the known
\ML-Rendezvous landscape.
Tables~\ref{tab:VD} and~\ref{tab:FUD}
make this relationship explicit. In these tables, SS, QS, and NS mean self-stabilizing,
quasi-self-stabilizing, and non-quasi-self-stabilizing
algorithms, respectively.
A single number \(k\) indicates that the exact color complexity
is \(k\).
An entry \(k\text{--}\ell\) indicates a lower bound of \(k\)
and an upper bound of \(\ell\), whereas \(k\text{--}\mathord{?}\)
indicates that the current lower bound is \(k\) but no matching
finite upper bound is known.
The term \emph{impossible} means that no deterministic
algorithm exists in that setting, regardless of the number of
colors.

For the intermediate atomic variants of \ASY, which are
not included in Table~\ref{tab:FUD}, we additionally obtain
four-color non-quasi-self-stabilizing algorithms for
\(\FS\) under M-atomic \ASY\ and for
\(\FC\) under CM-atomic \ASY, both under VD+FUD and
Rigid movement.

\begin{table}[!t]
  \centering
  \caption{Known and new results under Variable Disorientation (VD).
  \textcolor{red}{Red entries} indicate the results established in this paper. (*) This entry is the only non-\ML-algorithm~\cite{HDT}; all other algorithmic entries are \ML-algorithms~\cite{FSVY}.}
  \includegraphics[width=0.76\linewidth,trim=6 110 10 84,clip]{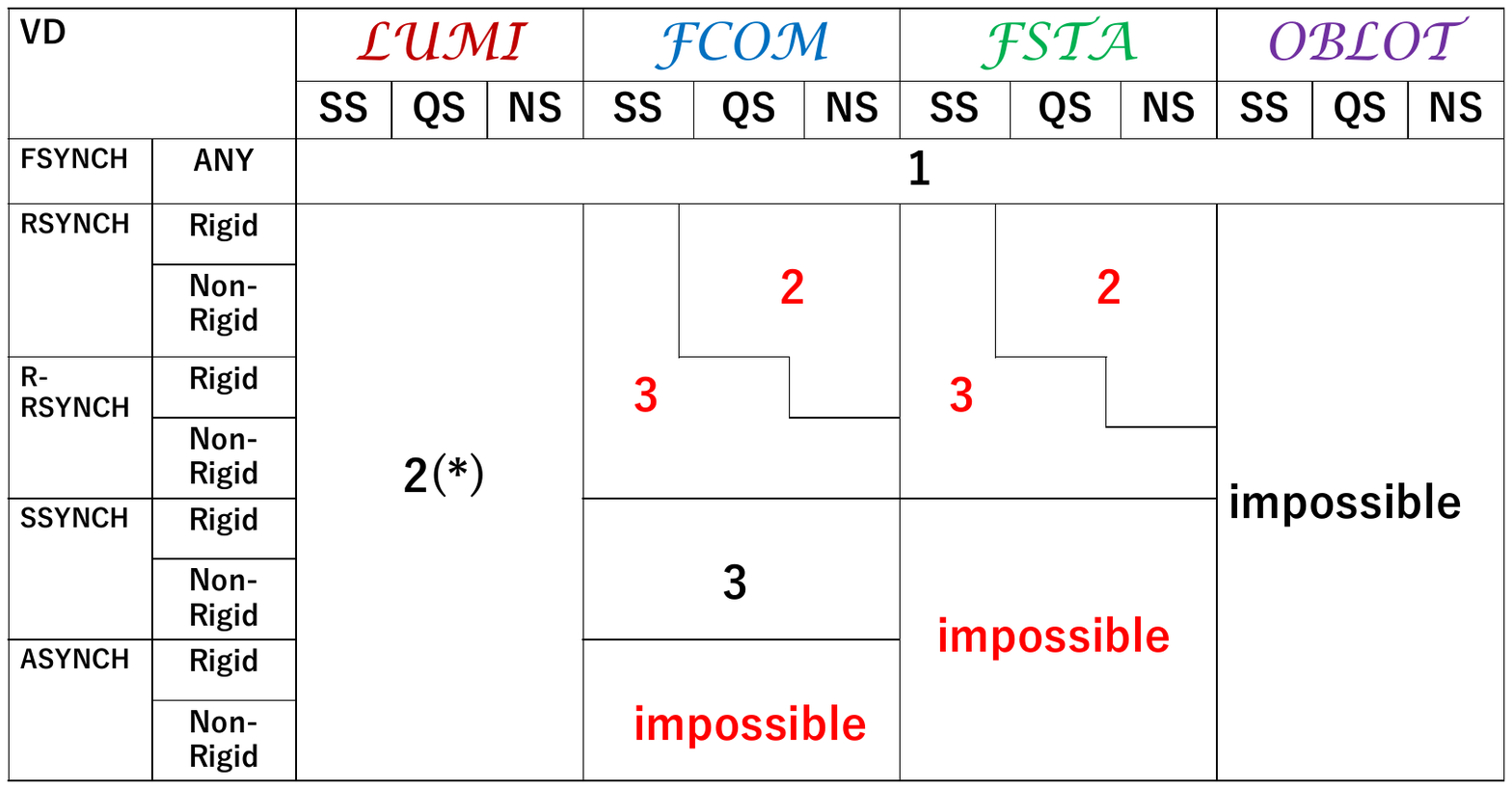}
    \label{tab:VD}
\end{table}

\begin{table}[!t]
  \centering
  \vspace{-0.1em}
  \caption{Known and new results under VD+Fixed-Unit-Distance
(FUD) for the standard schedulers.
Red numbers indicate exact color complexities or new upper
bounds established in this paper.
In the Rigid \(\FC\)-\ASY-NS setting, our six-color
algorithm improves the previous twelve-color upper
bound~\cite{FSVY}.}
  \includegraphics[width=0.76\linewidth,trim=5 80 10 82,clip]{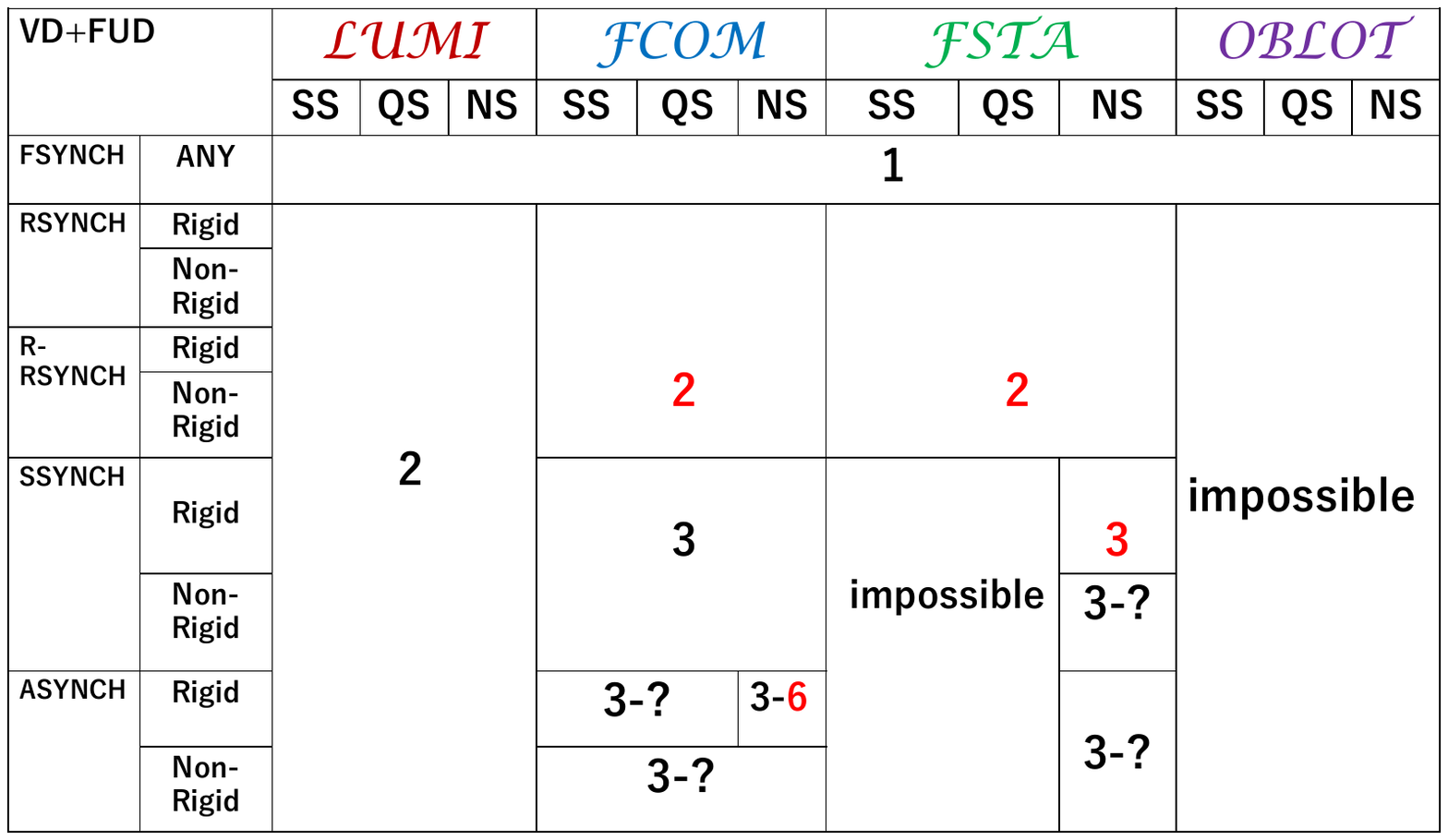}
    \label{tab:FUD}
\end{table}

\noindent\textbf{Our Contribution.}\enspace
Our results are organized around the contrast between VD and
VD+FUD, and around the use of a fixed local unit as a source of
persistent cyclic metric information.

First, under VD, we establish impossibility results for
\(\FS\) under \SSY\ and for \(\FC\) under \ASY,
without restricting the number of colors or the algorithm to
the \ML-class.
We also determine the role of two colors under \RSY\ and
R-\RSY\ in the one-sided visibility model.
Under \RSY, two colors suffice for quasi-self-stabilizing
\ML-Rendezvous but not for self-stabilization.
Under R-\RSY, no deterministic two-color
quasi-self-stabilizing algorithm exists, whereas under Rigid
movement a two-color non-quasi-self-stabilizing \(L\)-algorithm
does exist.

Second, under VD+FUD, we introduce periodic
\((b,m)\)-distance classes: logarithmic distance bands whose
successive boundaries differ by a factor \(b\), labeled
cyclically with \(m\) classes.
Consequently, multiplying or dividing the physical distance by
\(b\) shifts the class by one modulo \(m\), allowing a robot's
fixed local unit to serve as finite cyclic memory.
A related use of metric periodicity appears in a three-robot
Distinct Gathering algorithm under FUD, where a
golden-ratio-based periodic partition is employed~\cite{FreiWada26}.


The \((4,2)\)-classes break the mixed-color cycle and yield a
two-color self-stabilizing non-\ML-Rendezvous
algorithm for both \(\FS\) and \(\FC\) under \RSY\ and
R-\RSY, even with Non-Rigid movement.
Under Rigid movement, the \((2,2)\)- and \((2,3)\)-classes
yield an optimal three-color \(\FS\) algorithm under \SSY,
four-color algorithms for \(\FS\) under \(M\)-atomic \ASY\
and for \(\FC\) under \(CM\)-atomic \ASY, and a six-color
\(\FC\) algorithm under full \ASY, improving the previous
twelve-color upper bound to six.

\section{Robot Model}

\subsection{Robots and Configurations}

We consider two anonymous robots $r$ and $s$
operating in $\mathbb{R}^2$.
Robots execute deterministic \\
\Look–\Compute–\Move (\LCM) cycles~\cite{FPS}.
Robots have no identifiers and execute the same algorithm.
At Look, a robot records the peer's position and the light
information visible in its model.
At Compute, it determines its light update and movement;
at Move, it follows the computed trajectory.
Only light values persist between LCM cycles.
All algorithms presented below use straight-line movements.
Movements of robots are said to be {\em Rigid} if
the robots always reach the computed destination. 
They are said to be {\em Non-Rigid} if there exists a constant $\delta>0$ such that
whenever the computed destination is at distance at least $\delta$ from the
current position, the adversary may stop the robot before the destination
but only after it has moved at least $\delta$; if the destination is closer
than $\delta$, the robot reaches it.


Each robot has its own local coordinate system
in which it sees itself at the origin.
The local coordinate systems (i.e., the direction, the orientation, and the unit of distance) of different robots need not be consistent.
The absence of this consistency is called {\em disorientation}.
Under \emph{Variable Disorientation (VD)},
at every \Look phase the adversary may independently
apply an arbitrary rotation, reflection,
and scaling to each robot’s local coordinate system.
Hence,
directions are not globally consistent,
 orientation may flip, and
 distances are not metrically reliable.
Under \emph{VD+Fixed-Unit-Distance (FUD)},
the scaling factor of each robot’s local coordinate system
remains fixed throughout the execution.
Thus, while orientations may still vary,
each robot has a consistent notion of unit distance.
This additional geometric stability
strictly increases computational power.

Robots may be equipped with persistent lights whose joint
values are taken from a finite set \(\Gamma\) of colors.
A light update is committed atomically at the end of Compute,
before the corresponding Move begins.
The standard $\LU$ model allows a robot to observe both its own
light and the other robot's light, whereas $\OB$ robots have
no lights.
These two models appear in
Tables~\ref{tab:VD} and~\ref{tab:FUD} only to place our results
in the known landscape.

In this paper, we focus on the two one-sided visibility models:
\[
\begin{array}{ll}
\FS: & \text{a robot observes only its own light},\\
\FC: & \text{a robot observes only the other robot's light}.
\end{array}
\]
Thus, at each \Look, robot \(x\) observes exactly one color,
denoted by
\[
    \obs_x(t)=
    \begin{cases}
        c_x(t),&\text{in \FS},\\
        c_{\bar{x}}(t),&\text{in \FC},
    \end{cases}
\]
where $c_x(t)$ is the color of robot $x$ at time $t$ and $\bar{x}$ denotes the other robot.
Let \(p_x(t)\) be the position of robot \(x\) in a global
coordinate system unknown to the robots, and let
\(p_{\bar{x}}^x(t)\) denote the position of the other robot in
the local coordinate system of \(x\).
Since \(x\) sees itself at the origin, its relevant snapshot in
the one-sided visibility model can be represented compactly by
$
    SS_x(t)
    =
    \bigl(p_{\bar{x}}^x(t),\obs_x(t)\bigr).
$

A configuration at time \(t\) is
$
    \mathcal C(t)
    =
    \left\langle
        (p_r(t),c_r(t)),
        (p_s(t),c_s(t))
    \right\rangle.
$
Whenever an argument depends only on the fact that exactly one
color is visible, it applies uniformly to both $\FS$ and $\FC$.
In the general \(\FC\) model, a robot may update any subset of
its light components; all other components retain their current
values.
In both \(\FS\) and \(\FC\), \(\mathsf{Keep}\) denotes the
empty update and leaves the entire light unchanged.
All explicit algorithms in this paper use a single light
component, so \(\mathsf{Keep}\) is the only partial update they
use.

\subsection{Schedulers}

Executions are sequences of activations
chosen by an adversarial fair scheduler.

\textbf{\FSY:}
both robots are activated simultaneously in every round.

\textbf{\sc{Round Robin} (RR):}
exactly one robot is activated in each round, and the two
robots are activated alternately.

\textbf{\RSY:}
an execution is either a pure \FSY\ execution or consists of a
finite, possibly empty, prefix of \FSY\ rounds followed by an
infinite \RR suffix.

\textbf{R-\RSY\footnote{R means reverse.}:}
an execution is either a pure \RR\ execution or consists of a
finite, possibly empty, prefix of \RR\ rounds followed by an
infinite \FSY\ suffix.

\textbf{\SSY:}
in each round, a nonempty subset of robots is activated.

Thus, both \RSY\ and R-\RSY\ include pure \FSY\ and pure
\RR\ executions and are restricted subclasses of \SSY.
The two schedulers \RSY\ and R-\RSY\ are incomparable.

The following asynchronous schedulers differ in their
atomicity constraints.

\textbf{\(CM\)-atomic \ASY:}
once a robot begins its \Compute operation,
no other robot performs a \Look until that robot completes
the corresponding \Move operation.

\textbf{\(M\)-atomic \ASY:}
no robot performs a \Look while another robot is executing
its \Move operation.

\textbf{\ASY:}
robots execute LCM cycles asynchronously, with no global notion
of rounds and no atomicity constraint between their operations.

Every \(CM\)-atomic execution is \(M\)-atomic, and every
\(M\)-atomic execution is an \ASY\ execution.
For \FS, \(CM\)-atomic and \(M\)-atomic \ASY\ are
observationally equivalent, so the same algorithm with the same
number of colors works under both.
For \FC, they are observationally distinguishable but have the
same computational power; the known simulation may require
additional colors~\cite{FSSW23}.\footnote{The additional observable
configuration under \(M\)-atomicity occurs after the peer's
\Compute operation and before its \Move operation.
It is indistinguishable in \FS, because the peer's color is not
visible, but may reveal a new peer color in $\FC$.}

\subsection{Rendezvous and Stabilization}
\emph{Rendezvous} is achieved
if there is some time $t$ such that for any $t' \geq t$,
$p_r(t') = p_s(t')$.

Let \(\mathcal A\) be a Rendezvous algorithm.
It is \emph{self-stabilizing} if the two robots conforming to
\(\mathcal A\) solve Rendezvous from any configuration,
regardless of their initial colors.
It is \emph{quasi-self-stabilizing} if they solve Rendezvous
from any configuration in which the two robots have the same color.
It is \emph{non-quasi-self-stabilizing} if there exists a color
\(c\) such that they solve Rendezvous from any configuration
in which both robots have color \(c\).

A specific category of algorithms, \(\ML\)-algorithms,
requires each robot to determine both the movement parameter
\(\lambda\) and its light update solely from the colors
observed in the snapshot.
Neither may depend on the distance between the robots.
For the one-sided visibility model considered here, both
depend only on \(\obs_x(t)\).

\subsection{Algorithmic Notation and Periodic Distance Classes}
\label{subsec:algorithmic-notation}

This subsection introduces notation and geometric concepts used
in the descriptions and analyses of the algorithms presented later.
In particular, we define a unified representation of algorithms in
the one-sided visibility model, standard movement primitives, and
periodic distance classes induced by a robot's fixed unit distance.

\subsubsection{Algorithmic graphs and standard destinations}

Recall that, in the one-sided visibility model, each robot \(x\)
observes exactly one color, denoted by \(\obs_x(t)\).
In the \(\FS\) interpretation, \(\obs_x(t)\) is the color of \(x\)
itself, whereas in the \(\FC\) interpretation, it is the color of
the other robot \(\overline{x}\).

An algorithm in the one-sided visibility model can be represented
by a labeled directed graph
$
    G=(V,E,\ell),
$
where \(V\) is the set of colors, \(E\subseteq V\times V\) is the
set of color transitions, and the label of an edge specifies a
condition on the current snapshot and the corresponding destination.
For the single-light-component algorithms considered below, a rule
has the form
\(
    (a,\chi)
    \longmapsto
    (u,d),
\)
where \(a\) is the observed color, \(\chi\) is a condition on
the current snapshot,
\(
    u\in\Gamma\cup\{\mathsf{Keep}\}
\)
is the light update, and \(d\) is the destination.
If \(u\in\Gamma\), the robot writes \(u\) to its own light at
the end of \Compute.
If \(u=\mathsf{Keep}\), it performs no light update.

For a robot \(x\), a time \(t\), and a movement parameter
\(\lambda\in\mathbb{R}\), define
$
    \operatorname{dest}_x(t,\lambda)
    =
    (1-\lambda)p_x(t)+\lambda p_{\overline{x}}(t).
$
The algorithms presented later use the following three standard
destinations:
\[
\begin{aligned}
    \mathsf{Stay}_x(t)
        &=\operatorname{dest}_x(t,0)
          =p_x(t),\\
    \mathsf{Midpoint}_x(t)
        &=\operatorname{dest}_x\left(t,\frac12\right)
          =\frac{p_x(t)+p_{\overline{x}}(t)}{2},\\
    \mathsf{Other}_x(t)
        &=\operatorname{dest}_x(t,1)
          =p_{\overline{x}}(t).
\end{aligned}
\]
When \(x\) and \(t\) are clear from the context, we omit the
subscripts and arguments and simply write
\(\mathsf{Stay}\), \(\mathsf{Midpoint}\), and \(\mathsf{Other}\).
We also abbreviate these three destinations (or the corresponding
actions) by \(S\), \(H\), and \(O\), respectively, where
\(H\) stands for a \emph{halfway move}.

We use the full destination names in formal algorithm
descriptions and pseudocode, and use \(S,H,O\) in graph labels,
compact transition summaries, and geometric action pairs.

Under the \(\FS\) interpretation, the algorithmic graph is the
state-transition graph of each robot.  Under the \(\FC\)
interpretation, it is a reaction graph with respect to the color
observed on the other robot; in particular, it is not the
state-transition graph of the observing robot.

For synchronous executions, if a configuration \(C\) changes to
\(C'\) by activating the set
$
    a\in\bigl\{\{r\},\{s\},\{r,s\}\bigr\}
$
and letting every activated robot execute one complete cycle of an
algorithm \(\mathcal A\), we write
$
    C
    \xrightarrow[\mathcal A]{a}
    C'.
$
When the algorithm is clear from the context, the subscript
\(\mathcal A\) is omitted.  

\subsubsection{Physical and locally measured distances}
\label{subsubsec:locally_measured_distance}
Let
$
    dis(t)=\|p_r(t)-p_s(t)\|
$
denote the physical distance between the two robots at time \(t\).
For each robot \(x\), let \(u_x>0\) denote the physical length
corresponding to one unit in its local coordinate system.
The distance from \(x\) to the other robot measured in the local
coordinate system of \(x\) is denoted by
$    d_x(t)=\frac{dis(t)}{u_x}.$
Under FUD, \(u_x\) remains fixed throughout the execution,
although \(u_r\) and \(u_s\) need not be equal.

\subsubsection{Periodic distance classes}\label{subsubsec:periodic_distance_classes}

We next define a common framework for the distance partitions used
in several algorithms.

\begin{definition}[Periodic distance classes]
\label{def:periodic-distance-classes}
For \(b>1\), an integer \(m\geq2\), and a physical distance \(D>0\), define
\(
    \kappa_{x,b}^{(m)}(D)
    =
    \left\lfloor
        \log_b\frac{D}{u_x}
    \right\rfloor
    \pmod m.
\)
The adjective ``periodic'' refers to the fact that multiplying or
dividing a distance by the base \(b\) shifts its class cyclically.
\end{definition}

For a physical distance \(D>0\), we abbreviate
\(
    \kappa_x^{(2)}(D)
        :=\kappa_{x,2}^{(2)}(D),
    \:
    \kappa_x^{(3)}(D)
        :=\kappa_{x,2}^{(3)}(D).
\)
When the argument is a time \(t\), it denotes the class of the
current physical distance:
\(
    \kappa_{x,b}^{(m)}(t)
        :=\kappa_{x,b}^{(m)}(dis(t)).
\)
Whenever \(dis(t)=0\), every algorithm using periodic distance
classes applies the priority rule
\(
    (\mathsf{Keep},\mathsf{Stay})
\)
without evaluating a distance class.

\begin{lemma}[Periodic distance-class shift]
\label{lem:periodic-distance-class-shift}
Under FUD, for every integer \(h\geq0\),
\[
    \kappa_{x,b}^{(m)}(D/b^h)
    =
    \kappa_{x,b}^{(m)}(D)-h
    \pmod m.
\]
\end{lemma}

\begin{proof}
This follows immediately from the definition and the fact that
\(u_x\) is fixed under FUD.
\end{proof}

\begin{lemma}[Phase complementarity]
\label{lem:phase-complementarity}
Let the robots be at the same physical distance \(D>0\), and let
$
    i=\kappa_r^{(3)}(D),
    \:
    j=\kappa_s^{(3)}(D).
$
If robot \(r\) observes a phase color \(C_j\) and robot \(s\)
observes a phase color \(C_i\), then their phase offsets satisfy
$
    \Delta_r=i-j,
    \:
    \Delta_s=j-i=-\Delta_r
    \pmod3.
$
Consequently,
$
(\Delta_r,\Delta_s)
\in
\{(0,0),(1,2),(2,1)\}.
$
\end{lemma}

\section{Impossibility Results for Rendezvous}



\subsection{Impossibility under VD}

We first show that under VD,
even with rigid movements,
Rendezvous is impossible in certain scheduler/model combinations,
regardless of the number of colors
and without restricting to \ML-algorithms.

\begin{lemma}[VD Normalization]
\label{lem:vd-normalization}
Let $\C$ be any configuration in which the two robots occupy
distinct positions.
Under VD, at every \Look operation,
the adversary can ensure that each robot observes
the other at position $(0,1)$
in its local coordinate system.
\end{lemma}

\begin{proof}
Under VD, at every \Look,
the adversary may independently choose
the local coordinate system of the observing robot,
up to translation, rotation, reflection, and scaling.
Thus the observing robot can be placed at the origin,
the direction toward the other robot aligned with the positive $y$-axis,
and the distance scaled to $1$.
\end{proof}

The normalization lemma removes the essential distinction
between general Rendezvous algorithms and \ML-algorithms for
the purpose of the following impossibility argument.
Indeed, although a general algorithm may use the measured
distance to determine its destination and its light update,
under VD the adversary can normalize every non-Rendezvous
snapshot so that the other robot is always observed at
\((0,1)\).
Thus, as long as the robots remain distinct, the deterministic
behavior of a robot depends only on the observed color.

We can therefore lift the corresponding impossibility argument
for \ML-algorithms~\cite{FSVY} to arbitrary deterministic
algorithms.

\begin{theorem}\label{thm:onesided-impossible}
Under VD, there exists no deterministic
non-quasi-self-stabilizing Rendezvous algorithm, regardless of
the number of colors,
(i) for \(\FS\) under \SSY, or
(ii) for \(\FC\) under \ASY.
\end{theorem}

\begin{proof}[Proof sketch]
By Lemma~\ref{lem:vd-normalization}, from every non-Rendezvous
same-color configuration the adversary can make both robots
obtain the same normalized geometric view as in the
corresponding \ML-algorithm impossibility argument.
If the resulting move is not to the midpoint, a simultaneous
activation preserves distinct positions and produces another
same-color configuration.
If it is a midpoint move, the adversary uses sequential
activations under \SSY\ for \(\FS\), and a stale-snapshot
execution under \ASY\ for \(\FC\), to obtain another
non-Rendezvous same-color configuration.
Repeating the construction gives a fair infinite execution
without Rendezvous. The  detailed proof is deferred to Appendix~\ref{ap:thm:onesided-impossible}.
\end{proof}

Under VD, no Rendezvous algorithm exists for $\FC$ robots in \ASY,
even without restricting the algorithm to be in \ML.
In contrast, in the Non-Rigid \SSY\ model,
there exists a self-stabilizing \ML-Rendezvous algorithm
with $3$ colors for $\FC$ robots,
and this algorithm works even under VD~\cite{DFPSY}.



%
The following known lower bounds under \SSY\ hold even under
Fixed Disorientation, common chirality, and a common unit
distance~\cite{OFW}.
\begin{theorem}[{\cite{OFW}}]
\label{thm:ssynch-known-lower-bounds}
Under \SSY\ and Rigid movement, even with
Fixed Disorientation, common chirality, and a common unit
distance:
\begin{enumerate}
\item no two-color non-quasi-self-stabilizing Rendezvous
      algorithm exists in either \(\FS\) or \(\FC\); and
\item no quasi-self-stabilizing Rendezvous algorithm exists
      in \(\FS\), regardless of the number of colors.
\end{enumerate}
\end{theorem}


Consequently, any positive Rendezvous result for \(\FS\)
under \SSY\ in our VD+FUD setting must rely on a designated
initial color.
Combined with
Theorem~\ref{th:NS-FSTA-SSY-3color-optimal},
this yields a sharp separation:
quasi-self-stabilization is impossible with any number of
colors, whereas non-quasi-self-stabilizing Rendezvous is
solvable with exactly three colors.


\section{Two-Color Impossibility under
\RSY\ and
R-\RSY}
\label{sec:two-color-imp}

We establish two-color lower bounds under VD for the one-sided visibility
model with the schedulers  \RSY\ and R-\RSY.
By Lemma~\ref{lem:vd-normalization}, as long as the robots are
at distinct positions, the deterministic behavior of a robot
depends only on its observed color $c$.
We call \(c\) a \emph{midpoint color} if the destination
prescribed upon observing \(c\) is the midpoint, and a
\emph{move-to-other color} if it is the current position of
the peer.


\begin{theorem}
\label{thm:rsynch-onesided}
Under VD and \RSY, there exists no deterministic
two-color self-stabilizing Rendezvous algorithm in either
the \(\FS\) or the \(\FC\) model.
\end{theorem}
\begin{proof}[Proof sketch]
Pure \FSY\ requires a midpoint color, while pure \RR\ requires
a move-to-other color.
With only two colors, their transitions are forced:
the midpoint color must change to the move-to-other color and
vice versa.
Starting from a mixed-color configuration, a pure \FSY\
execution then preserves mixed colors and halves the distance
in every round.
Hence Rendezvous is never attained in finite time.
The complete proof is given in
Appendix~\ref{ap:thm:rsynch-onesided}.
\end{proof}

The same forced two-color structure also yields the corresponding
lower bounds for R-\RSY.
An \RR-prefix can create a non-gathered mixed-color
configuration, after which the \FSY\ tail maintains a
positive distance indefinitely.
In the Non-Rigid case, the adversary can create such a mixed
configuration even from the designated same-color initial
configuration.
The complete proof is given in
Appendix~\ref{ap:thm:rrsynch-fsfc}.

\begin{theorem}
\label{thm:rrsynch-fsfc}
Under VD and R-\RSY, there exists no deterministic
two-color quasi-self-stabilizing Rendezvous algorithm in either
the \(\FS\) or the \(\FC\) model.
Moreover, in the Non-Rigid setting, there exists no deterministic
two-color non-quasi-self-stabilizing Rendezvous algorithm in
either model.
\end{theorem}

%
%
%

\section{Algorithms in the One-Sided Visibility Model under \RSY\ and R-\RSY}

\subsection{Reduction Principles}
\label{subsec:rsynch-reduction}

We use two reduction principles whose formal statements and
proofs are given in Appendix~\ref{ap:rsynch-reduction}.
First, correctness under both \FSY\ and \RR\ yields correctness
under \RSY\ and R-\RSY, with the standard additional condition
for quasi-self-stabilization.
Second, for the algorithms considered below, the Non-Rigid
distance eventually becomes at most \(\delta\).
It therefore suffices to verify their finite-state behavior in
the effective Rigid regime \(D\leq\delta\).

\subsection{A Two-Color Quasi-Self-Stabilizing
\texorpdfstring{\ML-Algorithm}{L-Algorithm}}
\label{subsec:QS-L-2-color}

\begin{figure}
\centering\includegraphics[keepaspectratio, width=0.8 \textwidth]{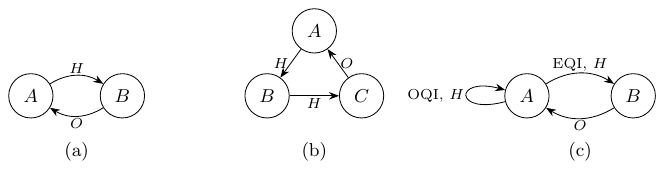}

      \caption{(a) Algorithm QS-\ML-2-color. (b) Algorithm SS-\ML-3-color. (c) Algorithm SS-non-\ML-2-color. EQI and OQI denote the even and odd periodic $(4,2)$-distance classes,
respectively.
      }\label{fig:all-algorithms}
\end{figure}

Theorem~\ref{thm:rsynch-onesided} shows that two colors are
insufficient for self-stabilization under \RSY.
We now show that two colors nevertheless suffice for
quasi-self-stabilization.

The algorithm \textsc{QS-\ML-2-color}, shown in
Figure~\ref{fig:all-algorithms}(a), uses the rules
$
    A\longmapsto(B,\mathsf{H}),
    \:\\
    B\longmapsto(A,\mathsf{O}).
$ The pseudocode for \textsc{QS-\ML-2-color} is shown in Appendix~\ref{ap:QS-L-2-color}.
The graph is interpreted with respect to the observed color
and therefore applies to both \FS\ and \FC.

\begin{theorem}
\label{th:QS-L-non-rigid-2-color}
\textsc{QS-\ML-2-color} is a two-color
quasi-self-stabilizing \ML-Rendezvous algorithm for both
\FS\ and \FC\ under \RSY\ in the Non-Rigid setting.
\end{theorem}

\begin{proof}[Proof sketch]
Under \FSY, a same-color configuration remains same-colored:
from \((A,A)\), both robots move to the midpoint, whereas from
\((B,B)\), they first exchange positions and then move to the
midpoint.
Thus no mixed-color configuration is generated.

Under \RR, once the distance is at most \(\delta\), Rendezvous
is achieved from any color configuration within at most three
activations in \FS\ and at most two activations in \FC.

For larger distances, the standard distance-decreasing argument
shows that every Non-Rigid execution eventually reaches
distance at most \(\delta\).
Combining the \FSY\ and \RR\ executions then gives the result
under \RSY.

Details are given in Appendix~\ref{ap:QS-L-2-color-correctness}.
\end{proof}

\begin{corollary}
\label{cor:NS-L-rigid-2-color-rrsynch}
Under R-\RSY\ and Rigid movement,
\textsc{QS-\ML-2-color}, with \(B\) as the designated initial
color, is a two-color non-quasi-self-stabilizing
\ML-Rendezvous algorithm for both \FS\ and \FC.
\end{corollary}

\begin{proof}
If the execution begins with an \RR\ activation, the activated
robot observes \(B\) and reaches the stationary peer.
If it begins with an \FSY\ round, the robots exchange positions
and enter \((A,A)\), after which the next \FSY\ round moves both
robots to the midpoint.
\end{proof}

\subsection{An Optimal Three-Color Self-Stabilizing
\texorpdfstring{\ML-Algorithm}{L-Algorithm} for
\texorpdfstring{\(\FS\)}{FSTA}}
\label{sec:SS-L-3-color}
The algorithm \textsc{SS-\ML-3-color}, shown in
Figure~\ref{fig:all-algorithms}(b), uses
$
    A\longmapsto(B,\mathsf{H}),
    \:
    B\longmapsto(C,\mathsf{H}),
    \:\\
    C\longmapsto(A,\mathsf{O}).
$ The pseudocode for \textsc{SS-\ML-3-color} is shown in Appendix~\ref{ap:SS-L-3-color}.

In the \FS\ interpretation, the three colors form a private
cycle.
Under \RR, every robot reaches color \(C\) after at most two
of its own activations and then performs a move-to-other
action.
Under \FSY, all color pairs are resolved in at most three
rounds once the distance is at most \(\delta\).

\begin{theorem}
\label{th:SS-L-non-rigid-3-color}
\textsc{SS-\ML-3-color} is a three-color self-stabilizing
\ML-Rendezvous algorithm for \FS\ robots under both
\RSY\ and R-\RSY\ in the Non-Rigid setting.
\end{theorem}

\begin{proof}[Proof sketch]
Correctness under \FSY\ follows by examining the nine possible
ordered color pairs; correctness under \RR\ follows from the
private color progression
$
    A\longrightarrow B\longrightarrow C\longrightarrow A.
$
The standard distance-decreasing argument extends these
effective-Rigid executions to Non-Rigid movement, and the
\FSY\ and \RR\ properties imply correctness under both
\RSY\ and R-\RSY.
The complete proof is given in Appendix~\ref{app:SS-L-3-color-correctness}.
\end{proof}

\begin{corollary}
Under VD and \RSY, the self-stabilizing color complexity of
Rendezvous for \FS\ robots is exactly three.
\end{corollary}

\begin{proof}
The upper bound follows from
Theorem~\ref{th:SS-L-non-rigid-3-color}, and the two-color
lower bound follows from
Theorem~\ref{thm:rsynch-onesided}.
\end{proof}

\begin{remark}[The \(\FC\) case]
\label{rem:SS-L-3-color-FCOM}
For \(\FC\), the exact three-color bound under \RSY\ already
follows from the known \SSY\ upper bound~\cite{FSVY} and
Theorem~\ref{thm:rsynch-onesided}.
The graph of Figure~\ref{fig:all-algorithms}(b) also works
under both \RSY\ and R-\RSY (Appendix~\ref{app:SS-L-3-color-correctness})
but not under \SSY; its main significance is therefore the
\(\FS\) separation between \RSY\ and \SSY.
\end{remark}


Combining the lower bounds of Section~\ref{sec:two-color-imp}
with the upper bounds above determines all red \RSY\ and
R-\RSY\ entries of Table~\ref{tab:VD} for both \FS\ and \FC.
In the order
\((\mathrm{SS},\mathrm{QS},\mathrm{NS})\), the exact color
complexities are
\begin{itemize}
\item   $ (3,2,2)
    \:\text{under \RSY,}$
\item $
    (3,3,2)
    \:\text{under R-\RSY\ with Rigid movement,}
$
and
\item 
$
    (3,3,3)
    \:\text{under R-\RSY\ with Non-Rigid movement.}
$
\end{itemize}
This gives the staircase pattern shown in
Table~\ref{tab:VD}.

\subsection{FUD Breaks the Two-Color
\texorpdfstring{\ML-Barrier}{L-Barrier}}
\label{sec:SS-non-L-2-color}

The previous results identify a two-color barrier within the
\ML\ framework.
Indeed, the forced-transition argument of
Theorem~\ref{thm:rsynch-onesided}, when restricted to
\ML-algorithms, does not require VD normalization:
a two-color self-stabilizing \ML-algorithm must combine a
midpoint color and a move-to-other color, and the resulting
mixed-color \FSY\ execution halves the distance forever.
We now show that FUD allows this barrier to be overcome once
the \ML\ restriction is dropped.
The key is to use the locally measured distance as an
additional cyclic bit.

We use the periodic \((4,2)\)-distance classes introduced in
Definition~\ref{def:periodic-distance-classes}.
For robot \(x\), let
$
    q_x(t)=\kappa_{x,4}^{(2)}(t).
$
By Lemma~\ref{lem:periodic-distance-class-shift},
$
    q_x(D/4)=q_x(D)-1\pmod2.
$
Thus, after two successive midpoint moves reduce the distance
from \(D\) to \(D/4\), the locally stored distance bit is
reversed.

The algorithm \textsc{SS-non-\ML-2-color}, shown in
Figure~\ref{fig:all-algorithms}(c), modifies
\textsc{QS-\ML-2-color} only in the rule for observed color
\(A\).
The robot writes either \(B\) or \(A\), depending on its
current periodic distance class.
The pseudocode for \textsc{SS-non-\ML-2-color} is shown in Appendix~\ref{ap:SS-non-L-2-color}.
\begin{theorem}
\label{th:SS-non-L-non-rigid-2-color}
Under VD+FUD,
\textsc{SS-non-\ML-2-color} is a two-color self-stabilizing
non-\ML-Rendezvous algorithm for both \FS\ and \FC\ under
\RSY\ and R-\RSY\ in the Non-Rigid setting.
\end{theorem}

\begin{proof}[Proof sketch]
The algorithm satisfies the reduction principles proved in
Appendix~\ref{ap:rsynch-reduction}.
In particular, under \FSY, simultaneous move-to-other actions
can occur only from  a \((B,B)\) configuration, which changes to
\((A,A)\) after one round, while under \RR\ every non-null
activation decreases the distance whenever it exceeds
\(\delta\).
Hence every Non-Rigid execution eventually reaches the
effective-Rigid regime in which the distance is at most
\(\delta\).

It therefore suffices to establish self-stabilization under
\FSY\ and \RR\ in this regime.
Correctness under both component schedulers then implies
correctness under \RSY\ and R-\RSY.


Same-observed-color configurations are resolved exactly as in
\textsc{QS-\ML-2-color}.
The only obstruction is a mixed-color configuration.
If the robot observing \(A\) sees an odd periodic distance
class, the execution reaches a same-observed-color
configuration after one phase.
Otherwise, if Rendezvous has not yet occurred after two
successive midpoint phases, the physical distance has changed
from \(D\) to \(D/4\).
The periodic class is therefore reversed, and the execution
enters the preceding case.
Thus the mixed-color cycle cannot repeat indefinitely.

The argument applies to both \FS\ and \FC\ because the graph is
defined solely in terms of the observed color and the
activated robot's own periodic distance class.
The complete \FSY\ and \RR\ analyses are given in
Appendix~\ref{ap:Proof-SS-non-L-2-color}.
\end{proof}

\section{Non-Quasi-Self-Stabilizing Rendezvous under VD+FUD and Rigid Movement}


\begin{figure}[t]
\centering
\includegraphics[width=\linewidth]{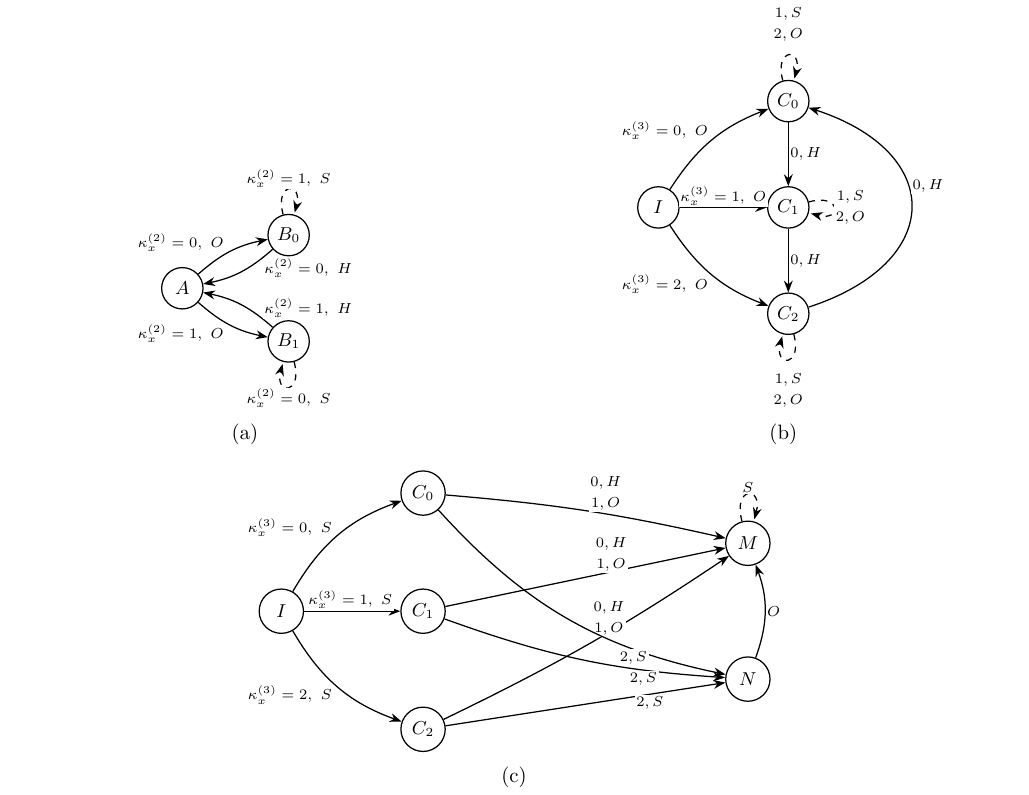}
\caption{
(a) Algorithm \textsc{NS-FSTA-3-color}.
(b) Algorithm \textsc{NS-FSTA-4-color}.
(c) Algorithm \textsc{NS-FCOM-6-color}.
In (b) and (c), a label \(i,X\) means that, when
$
    \Delta=\kappa_x^{(3)}(t)-j \pmod 3
$
equals \(i\), the robot performs action \(X\), where \(C_j\)
is the observed color. Dashed self-loops denote \(\mathsf{Keep}\), that is, the empty
light update.
}
\label{fig:new-algorithm-graphs}
\end{figure}

\subsection{A Three-Color Rendezvous Algorithm for \texorpdfstring{$\FS$}{FSTA}
under \texorpdfstring{\SSY}{SSYNCH}}
\label{subsec:FSTA-SSY-3color}

We now consider $\FS$ robots under \SSY, VD+FUD, and Rigid
movement.  We present a non-quasi-self-stabilizing Rendezvous
algorithm using three colors and show that three colors are optimal.

The algorithm, denoted by \textsc{NS-FSTA-3-color}, uses the
three colors
$
    \{A,B_0,B_1\},
$
where \(A\) is the designated initial color.
Since the robots are in $\FS$, each robot observes its own color.

The algorithm uses the periodic distance classes
\(\kappa_x^{(2)}\) defined in
Section~\ref{subsubsec:periodic_distance_classes}.
Recall that
$
    \kappa_x^{(2)}(D)
    =
    \kappa_{x,2}^{(2)}(D)
    \in\{0,1\},
$
and, by Lemma~\ref{lem:periodic-distance-class-shift},
\begin{equation}
\label{eq:3color-class-flip}
    \kappa_x^{(2)}(D/2)
    =
    1-\kappa_x^{(2)}(D).
\end{equation}

The behavior of a robot \(x\) is as follows.

\begin{itemize}
\item If \(x\) has color \(A\), let
$
    i=\kappa_x^{(2)}(t).
$
Robot \(x\) changes its color to \(B_i\) and moves to
\(\mathsf{Other}_x(t)\).

\item If \(x\) has color \(B_i\) and
$
    \kappa_x^{(2)}(t)=i,
$
then \(x\) changes its color to \(A\) and moves to
\(\mathsf{Midpoint}_x(t)\).

\item If \(x\) has color \(B_i\) and
$
    \kappa_x^{(2)}(t)\neq i,
$
then \(x\) keeps color \(B_i\) and stays at its current position.
\end{itemize}
The graph representation is shown in Figure~\ref{fig:new-algorithm-graphs}~(a) and the pseudocode is given in
Appendix~\ref{app:NS-FSTA-3color}.

In an algorithmic diagram, a solid edge from \(a\) to \(b\)
denotes a rule that observes \(a\) and writes \(b\) to the
robot's own light.
A dashed self-loop denotes a \(\mathsf{Keep}\)-rule; its label
specifies only the guard and destination.
Thus, a solid self-loop \(a\to a\) denotes an explicit write
of \(a\), whereas a dashed self-loop leaves the light unchanged.


Thus, color \(B_i\) records the periodic distance class observed
by the robot at its preceding activation in color \(A\).
A \(B_i\)-robot moves to the midpoint only as long as the current
distance remains in the recorded class.  Once the physical distance
is halved, Equation~\eqref{eq:3color-class-flip} makes the recorded
and current classes different, forcing the robot to stay.

\begin{theorem}
\label{th:NS-FSTA-SSY-3color-optimal}
Under VD+FUD and Rigid movement,
\textsc{NS-FSTA-3-color} solves Rendezvous for two $\FS$
robots under \SSY\ from the designated initial color \(A\).
Consequently, the non-quasi-self-stabilizing color complexity
of Rendezvous in this setting is exactly three.
\end{theorem}

\begin{proof}[Proof sketch]
Suppose first that the robots occupy distinct positions at
physical distance \(D>0\), initially with color \(A\).

If exactly one robot is activated in the first \SSY\ round,
its \(A\)-rule makes it move to the stationary peer.
Rigid movement therefore achieves Rendezvous immediately.

Suppose instead that both robots are activated.
Let
$
    i=\kappa_r^{(2)}(D),
    \:
    j=\kappa_s^{(2)}(D).
$
The robots change their colors to \(B_i\) and \(B_j\),
respectively, and exchange their positions.
The physical distance remains \(D\), so each stored index still
agrees with the corresponding robot's current distance class.

If both robots are activated in the next round, both move to
the same midpoint and achieve Rendezvous.
Otherwise, let \(r\) be the first robot activated alone.
Robot \(r\) changes its color from \(B_i\) to \(A\), moves to
the midpoint, and reduces the distance to \(D/2\).

By Lemma~\ref{lem:periodic-distance-class-shift},
$
    \kappa_s^{(2)}(D/2)
    =
    1-\kappa_s^{(2)}(D)
    =
    1-j
    \neq j.
$
Thus, while the distance remains \(D/2\), robot \(s\), which
still has color \(B_j\), stays whenever it is activated.
Robot \(r\), now in color \(A\), moves to the current position
of \(s\) at its next activation.
This remains true even if both robots are activated together,
because \(s\)'s action is null.
By fairness, \(r\) is eventually activated, and by rigidity it
reaches \(s\).

The argument is symmetric if \(s\) is activated first after the
position exchange.
If the robots initially occupy the same position, Rendezvous
has already been achieved.
Hence \textsc{NS-FSTA-3-color} solves Rendezvous from the
designated initial color \(A\).

A complete case analysis is given in
Appendix~\ref{ap:NS-FSTA-3-color-correctness}.

Finally, the two-color impossibility of Theorem~\ref{thm:ssynch-known-lower-bounds}~(1) gives the
matching lower bound.
Therefore, three colors are both necessary and sufficient.
\end{proof}

The same graph does not work for \(\FC\) under general FUD,
because the robots may assign different class indices to the
same physical distance.
However, it does work under $CM$-atomic \ASY\ when the robots agree on
their unit distance.
A counterexample and the positive argument are given in
Appendix~\ref{ap:NS-FSTA-3-color-on-FCOM}.

\subsection{Four-Color Phase Algorithms under Atomic Asynchrony}
\label{subsec:four-color-atomic}

We next present two four-color non-quasi-self-stabilizing
Rendezvous algorithms based on the same periodic three-phase
mechanism.
The first, \textsc{NS-FSTA-4-color}, is for \FS\ robots under
\(M\)-atomic \ASY, whereas the second,
\textsc{NS-FCOM-4-color}, is for \FC\ robots under
\(CM\)-atomic \ASY.

Both algorithms use the color set
$
    \{I,C_0,C_1,C_2\},
$
where \(I\) is the designated initial color, and the periodic
distance classes
$
    \kappa_x^{(3)}(t)
    =
    \kappa_{x,2}^{(3)}(t)
    \in\mathbb Z_3.
$
Recall from
Lemma~\ref{lem:periodic-distance-class-shift} that
$
    \kappa_x^{(3)}(D/2)
    =
    \kappa_x^{(3)}(D)-1
    \pmod3.
$

The algorithms have the same initialization rule.
When robot \(x\) observes \(I\), it sets
$
    i=\kappa_x^{(3)}(t),
$
changes its color to \(C_i\), and moves to the current position
of the other robot.

When \(x\) observes \(C_j\), it computes
$
    \Delta_x(t,j)
    =
    \kappa_x^{(3)}(t)-j
    \pmod3.
    $
The two algorithms then apply the following rules, where 
each entry specifies the light update and the geometric action:

\[
\begin{array}{c|c|c|c}
\text{observed color}
&
\text{condition}
&
\textsc{NS-FSTA-4-color}
&
\textsc{NS-FCOM-4-color}
\\ \hline
I
&
--
&
\bigl(C_{\kappa_x^{(3)}(t)},\mathsf{Other}_x(t)\bigr)
&
\bigl(C_{\kappa_x^{(3)}(t)},\mathsf{Other}_x(t)\bigr)
\\ \hline
C_j
&
\Delta_x(t,j)=0
&
(C_{j+1},\mathsf{Midpoint}_x(t))
&
(C_{j+1},\mathsf{Midpoint}_x(t))
\\
C_j
&
\Delta_x(t,j)=1
&
(\mathsf{Keep},\mathsf{Stay}_x(t))
&
(\mathsf{Keep},\mathsf{Other}_x(t))
\\
C_j
&
\Delta_x(t,j)=2
&
(\mathsf{Keep},\mathsf{Other}_x(t))
&
(\mathsf{Keep},\mathsf{Stay}_x(t))
\end{array}
\tag{\(\mathcal A_4\)}
\label{tab:four-color-phase-rules}
\]
Here and throughout this subsection, all phase indices are
taken modulo \(3\).

The opposite assignments of the two nonzero offsets reflect
the different meanings of a phase color.
In the \FS\ algorithm, \(C_j\) stores a distance class
previously measured by the same robot.
After a one-sided midpoint move, the mover obtains offset \(1\)
and must wait, whereas the peer obtains offset \(2\) and
approaches it.
In the \FC\ algorithm, \(C_j\) communicates a distance class
measured by the robot displaying it.
After a one-sided midpoint move, the mover observes the peer's
old phase and obtains offset \(2\), whereas the peer observes
the mover's advanced phase and obtains offset \(1\).
The nonzero actions are therefore reversed so that the midpoint
mover is again the waiter.

Figure~\ref{fig:new-algorithm-graphs}~(b) shows
\textsc{NS-FSTA-4-color}.
The graph of \textsc{NS-FCOM-4-color} is obtained by
replacing the loop labels \(1,S\) and \(2,O\) with
\(1,O\) and \(2,S\), respectively.
The corresponding pseudocodes are given in
Appendices~\ref{app:NS-FSTA-4-color}
and~\ref{app:NS-FCOM-4-color}, respectively.

\begin{theorem}[Four-color phase algorithms]
\label{th:four-color-phase-algorithms}
Under VD+FUD and Rigid movement, the following statements hold.
\begin{enumerate}
\item
\textsc{NS-FSTA-4-color} solves Rendezvous for two \FS\ robots
under \(M\)-atomic \ASY\ from the designated initial color \(I\).

\item
\textsc{NS-FCOM-4-color} solves Rendezvous for two \FC\ robots
under \(CM\)-atomic \ASY\ from the designated initial color
\(I\).
\end{enumerate}
\end{theorem}
\begin{proof}[Proof sketch]
If one robot completes its initial \(I\)-move before the peer
takes an \(I\)-snapshot, its Rigid move reaches the stationary
peer.
Otherwise, both initial snapshots precede the first move.
The two pending moves exchange the robots' positions and
produce a phase-consistent configuration at the unchanged
distance \(D\).
Any transient coincidence during this exchange is harmless,
because every activation at distance zero is null.

In both algorithms, offset \(0\) produces a midpoint move.
After a one-sided midpoint move, halving the distance decreases
the current class by one modulo \(3\).
In the \(\FS\) algorithm, the midpoint mover obtains offset
\(1\) and waits, whereas the peer obtains offset \(2\) and
approaches it.
In the \(\FC\) algorithm, the mover observes the peer's old
phase and obtains offset \(2\), whereas the peer observes the
mover's advanced phase and obtains offset \(1\).
The reversed nonzero mapping therefore again makes the midpoint
mover the waiter.

For \(\FC\), Lemma~\ref{lem:phase-complementarity} also shows
that every nonzero phase configuration yields an \(O/S\) or
\(S/O\) pair.
The respective atomicity assumptions make all pending and
repeated actions compatible with these geometric resolutions.
The complete asynchronous case analyses are given in
Appendix~\ref{aP:subsec:four-color-atomic}.
\end{proof}

\subsection{Extending the Four-Color Algorithm to Full Asynchrony}
\label{subsec:FCOM-ASYNCH-6color}
We extend \textsc{NS-FCOM-4-color} to full \ASY\ by adding
two guard colors, \(M\) and \(N\).
Under \(CM\)-atomic \ASY, the peer cannot observe a newly
published phase color before the associated move finishes.
Full \ASY\ permits such observations both before and during
the move, so non-null moves must be explicitly guarded.

The six-color algorithm retains the phase-to-action mapping
\(0\mapsto H\), \(1\mapsto O\), \(2\mapsto S\).
The first two rules write \(M\), announcing a committed
non-null move; the third writes \(N\), marking a stationary
anchor.
Observing \(M\) produces \((\mathsf{Keep},\mathsf{Stay})\),
whereas observing \(N\) makes the robot write \(M\) and move
to the anchor.
Thus every non-null move is performed while the mover displays
\(M\).

Initialization is passive: a robot observing \(I\) publishes
\(C_{\kappa_x^{(3)}(t)}\) and stays, instead of executing the
four-color algorithm's initial swap.
Retaining that swap would leave an unguarded move during which
the peer could combine a phase color with a different transient
distance.
Passive initialization keeps the physical distance unchanged
until the first \(M\) appears.
The colors \(C_0,C_1,C_2\) form the phase layer, with \(I\)
for initialization; \(M,N\) form the guard layer.

For a robot \(x\) observing \(C_j\), let
\(\Delta_x(t,j)=\kappa_x^{(3)}(t)-j\pmod3\).
The complete transition rules are as follows.

\begin{center}
\begin{tabular}{c|c|c|c}
observed color & condition & light update & destination\\ \hline
\(I\) & -- &
\(C_{\kappa_x^{(3)}(t)}\) & \(\mathsf{Stay}_x(t)\)\\ \hline
\(C_j\) & \(\Delta_x(t,j)=0\) &
\(M\) & \(\mathsf{Midpoint}_x(t)\)\\
\(C_j\) & \(\Delta_x(t,j)=1\) &
\(M\) & \(\mathsf{Other}_x(t)\)\\
\(C_j\) & \(\Delta_x(t,j)=2\) &
\(N\) & \(\mathsf{Stay}_x(t)\)\\ \hline
\(M\) & -- &
\(\mathsf{Keep}\) & \(\mathsf{Stay}_x(t)\)\\
\(N\) & -- &
\(M\) & \(\mathsf{Other}_x(t)\)
\end{tabular}
\end{center}

The graph representation is shown in
Figure~\ref{fig:new-algorithm-graphs}~(c), and the pseudocode
is given in Appendix~\ref{app:NS-FCOM-6-color}.

\begin{lemma}[Moving guard]
\label{lem:six-color-moving-guard}
Suppose that a robot \(r\) computes a non-null move in
\textsc{NS-FCOM-6-color}.
At the end of this \Compute operation, \(r\) displays color \(M\).

Until \(r\) completes the move, every fresh \Look performed by
the other robot \(s\) observes \(M\) and hence generates a
null move.
Consequently, while \(M\) remains visible, no fresh snapshot of
\(s\) can generate a new non-null move.
Robot \(s\) may have at most one incomplete cycle whose \Look
occurred before \(M\) became visible.
\end{lemma}

\begin{proof}
Every rule producing a non-null destination writes \(M\) to the moving robot's light at the end of \Compute. 
Since the rule for observed color \(M\) is
\(\mathsf{Stay}\), every subsequent fresh \Look of the peer
generates a null move.

Moreover, each robot executes its \LCM cycles sequentially.
Thus, when \(M\) becomes visible, the peer can have at most one
incomplete cycle based on an earlier snapshot.
\end{proof}

\begin{theorem}
\label{th:six-color-FCOM-ASYNCH}
Under VD+FUD and Rigid movement,
\textsc{NS-FCOM-6-color} solves Rendezvous for two \(\FC\)
robots under \ASY\ from the designated initial color \(I\).
Consequently, there exists a six-color
non-quasi-self-stabilizing Rendezvous algorithm in this setting.
\end{theorem}
\begin{proof}[Proof sketch]
Consider a fair execution starting from \((I,I)\) at positive
physical distance \(D\).

Before the first occurrence of \(M\), every completed action is
null: both the initialization rule and the offset-\(2\) phase
rule prescribe \(\mathsf{Stay}\).
Hence the distance remains \(D\), and every phase color
generated before the first \(M\) is consistent with this unchanged distance.

Fairness guarantees that \(M\) eventually appears.
After a phase color is published, offset \(0\) or \(1\)
produces \(M\), whereas offset \(2\) produces \(N\), after
which the peer eventually observes \(N\) and enters \(M\).
If both robots publish phase colors, phase complementarity
ensures that at least one has offset \(0\) or \(1\).

Let \(r\) be the first robot that enters \(M\) and computes a
non-null move.
By Lemma~\ref{lem:six-color-moving-guard}, until this move is
completed, every cycle of the peer \(s\) whose \Look occurs
after \(M\) becomes visible generates a null action.
At most one cycle of \(s\), whose \Look occurred before the
first \(M\), may remain incomplete.

Treating the absence of such a cycle as a null action, the
possible geometric action pairs are, up to symmetry,
$
    O/S,\: H/H,\: H/S.
$
Indeed, if the first \(M\)-move is triggered by observing
\(N\), it forms an \(O/S\) pair.
Otherwise, it is triggered by a phase color consistent with the unchanged distance.
The peer's possible pre-\(M\) cycle is either a null
initialization action, yielding \(H/S\) or \(O/S\), or another
phase-based action.
In the latter case, phase complementarity and the mapping
$
    0\mapsto H,\: 1\mapsto O,\: 2\mapsto S
$
yield \(H/H\), \(O/S\), or \(S/O\).

In the \(O/S\) case, the Rigid move reaches the stationary
robot.

In the \(H/H\) case, both actions have the same midpoint as
their destination.
After the first robot reaches it, every subsequent activation
of that robot is null: it either observes \(M\), or observes
the still-visible phase color at distance \(D/2\), obtains
offset \(2\), and enters \(N\).
It therefore remains at the midpoint until the other Rigid
midpoint move arrives.

In the \(H/S\) case, let \(C_j\) be the phase color still
displayed by the stationary peer.
After the midpoint move, the mover's distance class decreases
from \(j\) to \(j-1\), so its offset from \(C_j\) is \(2\).
It therefore enters \(N\) and stays.
The peer eventually observes \(N\), enters \(M\), and moves to
the stationary \(N\)-robot.
Lemma~\ref{lem:six-color-moving-guard} prevents fresh non-null
interference during this move, and rigidity completes
Rendezvous.

Thus every fair execution from \((I,I)\) reaches Rendezvous in
finite time, after which the zero-distance rule keeps the
robots together.
A complete event-level analysis, including all possible stale
cycles, is given in
Appendix~\ref{app:NS-FCOM-6-color-correctness}.
\end{proof}

\section{Concluding Remarks}
\label{sec:concluding-remarks}

Together with known results, our bounds complete the VD
classification in Table~\ref{tab:VD}.
Under VD+FUD, periodic distance classes turn a fixed local
metric into cyclic memory, yielding two-, three-, four-, and
six-color algorithms.
These gains require neither agreement on unit distance nor
stable orientations: the key geometric resource is the
consistency of each robot's own unit over time.
For full \ASY, the six-color construction combines this
metric phase mechanism with movement guards and passive
initialization to handle stale observations.

Open questions include closing the \(3\)--\(6\) gap for
non-quasi-self-stabilizing \(\FC\) Rendezvous under full
Rigid \ASY, resolving the remaining stabilization and
\(\FS\) cases, and extending the Rigid constructions to
Non-Rigid movement and larger robot populations.

\section*{Acknowledgment}
This work was supported by JSPS KAKENHI Grant Numbers
JP20KK0232, JP25K15057, 
JP25K03078, 
JP25K03079, 
JP26K02865, 
and by JST FOREST Program JPMJFR226U. 

\newpage
\bibliographystyle{alpha} 
\bibliography{referencesorg}

\clearpage
\appendix
\raggedbottom
%
%
\noindent
{\Large
{\bf Appendix}}
\section{Unified One-Sided Visibility Notation}\label{ap:0}
For each robot $x\in\{r,s\}$, let $\bar x$ denote the other robot.
In the one-sided visibility model, the observed color of robot $x$ at time $t$
is denoted by $\obs_x(t)$ and is defined as
\[
\obs_x(t)=
\begin{cases}
c_x(t) & \text{in }\FS,\\
c_{\bar x}(t) & \text{in }\FC.
\end{cases}
\]
Thus, the same algorithmic graph can be interpreted uniformly
with respect to the observed color.

\section{Color-Independent Impossibility under VD}\label{ap:thm:onesided-impossible}
\setcounterref{theoremduplicate}{thm:onesided-impossible}
\addtocounter{theoremduplicate}{-1}

\begin{theoremduplicate}
[Reprint of Theorem~\ref{thm:onesided-impossible} on 
page~\pageref{thm:onesided-impossible}]
\label{duplicate:thm:onesided-impossible}
Under VD, there exists no deterministic non-\QS\ Rendezvous algorithm  
regardless of the number of colors
(i) in the $\FS$ model under \SSY, or
(ii) in the $\FC$ model under \ASY.
\end{theoremduplicate}

\begin{proof}
Assume, for contradiction, that such an algorithm $\mathcal A$ exists.
Consider an initial non-gathered configuration
in which both robots have the same color.

By Lemma~\ref{lem:vd-normalization},
each robot always sees the other at $(0,1)$
in its own local coordinate system.
Hence, in both models, the two robots compute
the same light update and the same destination.

We show that from any non-gathered same-color configuration,
the adversary can construct a finite legal execution fragment
that ends in another non-gathered same-color configuration.
By repeating this fragment forever,
the adversary obtains an infinite fair execution
with no Rendezvous, yielding a contradiction.

\medskip
\noindent
\textbf{Case 1: the computed destination is not the midpoint.}

The adversary lets both robots perform one cycle
based on identical snapshots.

If the destination lies on the segment connecting the robots
but is not the midpoint,
then after the moves the robots remain distinct.
If the destination lies off the segment,
the adversary chooses the local coordinate systems
so that the two computed destinations become mirror images,
and hence are distinct as well.

Thus Rendezvous does not occur.
Since both robots executed the same rule,
they also obtain the same resulting color.
Therefore the resulting configuration
is again a non-gathered same-color configuration.

\medskip
\noindent
\textbf{Case 2: the computed destination is the midpoint.}

We distinguish the two models.

\smallskip
\noindent
\emph{Case 2a: $\FS$ under \SSY.}
The adversary activates the robots sequentially in two rounds.
First, one robot moves to the midpoint and changes its color.
Then the other robot, still having the original color,
observes the same situation and executes the same rule.
After these two activations,
the robots remain distinct but have the same new color.
Hence the resulting configuration is again
a non-gathered same-color configuration.

\smallskip
\noindent
\emph{Case 2b: $\FC$ under \ASY.}
The adversary executes the following asynchronous schedule:
\begin{enumerate}
\item activate robot $r$ for \Look only,
\item activate robot $s$ for two complete LCM cycles,
\item complete the pending \Compute--\Move of $r$.
\end{enumerate}
Since $r$ computes using a stale snapshot,
while $s$ performs two cycles observing the same color each time,
the robots end in distinct positions but with the same new color.
Hence the resulting configuration is again
a non-gathered same-color configuration.
Robot \(s\) applies the same componentwise update in both
cycles; since this update is idempotent, its second application
does not change the resulting color.
Robot \(r\) later applies the same update to the original
color, and hence both robots obtain the same resulting color.

\medskip
In all cases, the adversary can transform
any non-gathered same-color configuration
into another non-gathered same-color configuration.
Repeating this construction indefinitely
produces an infinite fair execution
in which Rendezvous never occurs,
contradicting the correctness of $\mathcal A$.

Therefore no deterministic non-quasi-self-stabilizing Rendezvous algorithm exists.
\end{proof}


\section{Two-Color Impossibility under \RSY}
\label{ap:thm:rsynch-onesided}

For each observed color \(c\in\{A,B\}\), let
\(
    U_c:\{A,B\}\longrightarrow\{A,B\}
\)
denote the map induced by the possibly partial light update
prescribed upon observing \(c\), and define
\(
    \tau_c(a)=U_c(a),
    \:
    \tau(c)=\tau_c(c).
\)
Thus, \(\tau_c(a)\) is the resulting color when a robot
currently displaying \(a\) applies the update prescribed for
observed color \(c\).

In the \(\FS\) model, the current and observed colors coincide,
so only the diagonal values \(\tau_c(c)\) can arise.
In the \(\FC\) model, off-diagonal values may also arise.
Since each \(U_c\) is induced by a componentwise partial
assignment, it is idempotent.
Consequently,

\begin{lemma}
\label{lem:two-color-update}
For each \(c\in\{A,B\}\), let \(\bar c\) denote the other
color. Then
\[
    \tau(c)=\bar c
    \quad\Longrightarrow\quad
    U_c=\mathrm{const}_{\bar c},
\]
whereas
\[
    \tau(c)=c
    \quad\Longrightarrow\quad
    U_c\in\{\mathrm{id},\mathrm{const}_c\}.
\]
Here, \(\mathrm{const}_c\) denotes the constant map with value
\(c\), and \(\mathrm{id}\) denotes the identity map on
\(\{A,B\}\).
\end{lemma}
\begin{proof}
If \(\tau(c)=\bar c\), then \(U_c(c)=\bar c\).
By idempotence,
\(
    U_c(\bar c)
    =
    U_c(U_c(c))
    =
    U_c(c)
    =
    \bar c,
\)
and hence \(U_c=\mathrm{const}_{\bar c}\).

If \(\tau(c)=c\), then \(U_c(c)=c\).
Since \(U_c(\bar c)\) is either \(c\) or \(\bar c\),
\(U_c\) is respectively \(\mathrm{const}_c\) or
\(\mathrm{id}\).
\end{proof}

The following is the core lemma to prove the two-color impossibility.

\begin{lemma}
\label{lem:two-color-core}
Let \(\mathcal A\) be a deterministic two-color Rendezvous
algorithm under VD.

Suppose that, for some color \(c^*\in\{A,B\}\),
\(\mathcal A\) solves Rendezvous under both pure \FSY\ and
pure \RR\ from every non-gathered configuration whose color
configuration is \((c^*,c^*)\).


Then one color must be a midpoint color and the other must
be a move-to-other color.
Denote them by \(A\) and \(B\), respectively.
With this naming, the following statements hold.
\begin{enumerate}
\item
If \(\mathcal A\) solves Rendezvous from \((A,A)\) under pure
\RR, then
\(
    \tau(A)=B
    \:\text{and}\:
    U_A=\mathrm{const}_B.
\)
\item
If \(\mathcal A\) solves Rendezvous from \((B,B)\) under pure
\FSY, 
then \(\tau(B)=A\) and \(U_B=\mathrm{const}_A\).
\item
If \(U_A=\mathrm{const}_B\) and
\(U_B=\mathrm{const}_A\), then from every non-gathered
mixed-color configuration there exists a pure \FSY\
execution, with every move completed, that never achieves
Rendezvous.
\end{enumerate}
\end{lemma}

\begin{proof}
Starting from a same-color configuration under pure \FSY, the
robots have the same current color, observe the same color, and
apply the same light update in every round.
Hence the configuration remains same-colored.
In the round in which Rendezvous is achieved, their common
observed color must prescribe a midpoint move: for any other
canonical destination, the adversary can choose their local
coordinate systems so that the two physical destinations are
distinct.
Thus at least one color must be a midpoint color.

Under pure \RR, only one robot moves in each round.
Hence finite-time Rendezvous requires some observed color for
which the activated robot moves exactly to the current position
of the stationary peer.
Thus at least one color must be a move-to-other color.
Since these two destinations are distinct at positive distance
and only two colors exist, denote the midpoint color by \(A\)
and the move-to-other color by \(B\).

For statement~1, suppose that \(\tau(A)=A\).
Starting from \((A,A)\) under pure \RR, every activated robot
remains in color \(A\) and moves to the midpoint.
Allowing every move to be completed, the distance is repeatedly
halved but never becomes zero in finite time.
Therefore \(\tau(A)=B\), and idempotence gives
\(U_A=\mathrm{const}_B\).

For statement~2, suppose that \(\tau(B)=B\).
Starting from \((B,B)\) under pure \FSY, both robots remain in
color \(B\) and repeatedly exchange positions.
Hence Rendezvous is never achieved.
Therefore \(\tau(B)=A\), and idempotence gives
\(U_B=\mathrm{const}_A\).

Finally, assume the hypotheses of statement~3 and consider a
pure \FSY\ execution starting from \((A,B)\), allowing every
move to be completed.
In the \(\FS\) model, the color configuration alternates
between \((A,B)\) and \((B,A)\).
In the \(\FC\) model, it remains \((A,B)\).
In either model, exactly one robot moves to the midpoint and
the other moves to its peer's current position in every round.
Thus the physical distance is halved in every round and remains
positive after every finite number of rounds.
Hence Rendezvous is never achieved.
\end{proof}

\setcounterref{theoremduplicate}{thm:rsynch-onesided}
\addtocounter{theoremduplicate}{-1}

\begin{theoremduplicate}
[Reprint of Theorem~\ref{thm:rsynch-onesided} on
page~\pageref{thm:rsynch-onesided}]
\label{duplicate:thm:rsynch-onesided}
Under VD and \RSY, there exists no deterministic
two-color self-stabilizing Rendezvous algorithm in either
the \(\FS\) or the \(\FC\) model.
\end{theoremduplicate}

\begin{proof}
Since \RSY\ contains both pure \FSY\ and pure \RR\
executions, a self-stabilizing algorithm must solve Rendezvous
from both \((A,A)\) and \((B,B)\) under both component
schedulers.

By Lemma~\ref{lem:two-color-core}, the colors can be named so
that \(A\) is a midpoint color and \(B\) is a move-to-other
color, and
\(
    U_A=\mathrm{const}_B,
    \qquad
    U_B=\mathrm{const}_A.
\)
Statement~3 of the same lemma then gives a pure \FSY\
execution from a mixed-color configuration in which the
distance remains positive at every finite time.
This is a legal \RSY\ execution, contradicting
self-stabilization.
\end{proof}


\section{Two-Color Impossibility under R-\RSY}
\label{ap:thm:rrsynch-fsfc}

We use the notation \(U_c\), \(\tau_c(a)\), and \(\tau(c)\)
introduced in Appendix~\ref{ap:thm:rsynch-onesided}.

\setcounterref{theoremduplicate}{thm:rrsynch-fsfc}
\addtocounter{theoremduplicate}{-1}

\begin{theoremduplicate}
[Reprint of Theorem~\ref{thm:rrsynch-fsfc} on
page~\pageref{thm:rrsynch-fsfc}]
\label{duplicate:thm:rrsynch-fsfc}
Under VD and R-\RSY,
there exists no deterministic two-color
quasi-self-stabilizing Rendezvous algorithm
in either the \(\FS\) or the \(\FC\) model.
Moreover, in the Non-Rigid model,
there exists no deterministic two-color
non-quasi-self-stabilizing Rendezvous algorithm
in either the \(\FS\) or the \(\FC\) model.
\end{theoremduplicate}

\begin{proof}
We again prove the statement in the unified one-sided
visibility model.

\smallskip
\noindent
\textbf{First claim: no two-color quasi-self-stabilizing
algorithm exists.}

Assume, for contradiction, that such an algorithm exists.
Since R-\RSY\ contains both pure \FSY\ and pure \RR\
executions and quasi-self-stabilization requires correctness
from every same-color configuration,
Lemma~\ref{lem:two-color-core} implies that the colors can be
named so that \(A\) is a midpoint color and \(B\) is a
move-to-other color, and
\(
    U_A=\mathrm{const}_B,
    \:
    U_B=\mathrm{const}_A.
\)

Consider an R-\RSY\ execution consisting of one \RR\
activation followed by an infinite \FSY\ suffix.
Start from \((A,A)\), and allow the first midpoint move to be
completed.
The activated robot changes its color to \(B\), so the
resulting configuration is mixed-colored and the robots remain
at distinct positions.
Statement~3 of Lemma~\ref{lem:two-color-core} shows that the
subsequent pure \FSY\ execution never achieves Rendezvous,
a contradiction.

\smallskip
\noindent
\textbf{Second claim: no two-color
non-quasi-self-stabilizing algorithm exists in the Non-Rigid
setting.}

Assume, for contradiction, that such an algorithm exists, and
let \(c^*\in\{A,B\}\) be its designated initial color.
Since it must solve Rendezvous from \((c^*,c^*)\) under both
pure \FSY\ and pure \RR, the first part of
Lemma~\ref{lem:two-color-core} implies that one color is a
midpoint color and the other is a move-to-other color.
Again denote them by \(A\) and \(B\), respectively.

We distinguish the designated initial color.

\smallskip
\noindent
\textbf{Case 1: \(c^*=A\).}

By statement~1 of Lemma~\ref{lem:two-color-core},
\(
    \tau(A)=B,
    \:
    U_A=\mathrm{const}_B.
\)
Consider one \RR\ activation from \((A,A)\), allow the
midpoint move to be completed, and then continue with an
infinite \FSY\ suffix.
The first activation produces a non-gathered mixed-color
configuration.  Allow every move in the subsequent \FSY\ suffix
to be completed.

If \(\tau(B)=A\), then \(U_B=\mathrm{const}_A\) by
idempotence, and statement~3 of
Lemma~\ref{lem:two-color-core} gives an infinite
non-Rendezvous execution.

Suppose instead that \(\tau(B)=B\).
Then \(U_B\) is either the identity map or the constant-\(B\)
map.
In the \(\FS\) model, the first \FSY\ round from the mixed
configuration produces \((B,B)\).
In the \(\FC\) model, the mixed configuration is preserved
when \(U_B=\mathrm{id}\), whereas \((B,B)\) is produced when
\(U_B=\mathrm{const}_B\).

If the configuration remains mixed, every subsequent
\FSY\ round consists of one midpoint move and one
move-to-other action, and the distance is halved forever.
If \((B,B)\) is reached, both robots retain color \(B\) and
repeatedly exchange positions in every \FSY\ round.
Thus Rendezvous is never achieved in either case.

\smallskip
\noindent
\textbf{Case 2: \(c^*=B\).}

By statement~2 of Lemma~\ref{lem:two-color-core},
\(
    \tau(B)=A,
    \:
    U_B=\mathrm{const}_A.
\)

Choose an initial physical distance \(D>4\delta\).
In the first \RR\ activation, one robot changes its color from
\(B\) to \(A\) and attempts to move to the other robot.
The Non-Rigid adversary stops this move after distance
\(\delta\).
Thus the robots remain distinct, and the color configuration
becomes mixed.

We distinguish the two visibility models.

\smallskip
\noindent
\emph{\(\FS\) interpretation.}

If \(\tau(A)=B\), then \(U_A=\mathrm{const}_B\), and
statement~3 of Lemma~\ref{lem:two-color-core} applies to an
infinite \FSY\ suffix.

Suppose instead that \(\tau(A)=A\).
Continue with pure \RR.
At the next activation, the robot still having color \(B\)
observes \(B\), changes its color to \(A\), and attempts to
move to the other robot.
The adversary again stops the move after distance \(\delta\).
The robots remain distinct and now have color configuration
\((A,A)\).

Thereafter, allow every move to be completed.
Every activated robot observes \(A\), remains in color \(A\),
and moves to the midpoint.
Thus the distance is halved at every activation but never
becomes zero in finite time.

\smallskip
\noindent
\emph{\(\FC\) interpretation.}

After the first \RR\ activation, the robot still displaying
\(B\) observes the color \(A\) displayed by the other robot.

If
\(
    \tau_A(B)=B,
\)
continue with an infinite \FSY\ suffix and allow every move
in the suffix to be completed.
In every round, the \(A\)-robot observes \(B\), remains in
color \(A\) because \(U_B=\mathrm{const}_A\), and moves to
the other robot.
The \(B\)-robot observes \(A\), remains in color \(B\), and
moves to the midpoint.
Hence the mixed-color configuration is preserved and the
distance is halved in every round.

Finally, suppose that
\(
    \tau_A(B)=A.
\)
Continue with pure \RR.
At the next activation, the \(B\)-robot observes \(A\),
changes its color to \(A\), and computes a midpoint move.
After the first interrupted move, the physical distance is
greater than \(3\delta\), so this midpoint is at distance
greater than \(\delta\).
The adversary can therefore stop the midpoint move after
distance \(\delta\), leaving the robots at distinct positions
with color configuration \((A,A)\).

By idempotence,
\(
    \tau_A(A)
    =
    U_A(A)
    =
    U_A(U_A(B))
    =
    U_A(B)
    =
    A.
\)
Hence every subsequent activated robot observes \(A\), remains
in color \(A\), and moves to the midpoint.
Allowing these moves to be completed, pure \RR\ repeatedly
halves the distance without making it zero in finite time.

In every case, we have constructed a fair legal Non-Rigid
R-\RSY\ execution starting from the designated same-color
configuration in which Rendezvous is never achieved.
This contradiction proves the second claim.
\end{proof}



\section{Scheduler Composition and Non-Rigid Distance Reduction}
\label{ap:rsynch-reduction}

Self-stabilization under both \FSY\ and \RR\ yields
self-stabilization under \RSY\ and R-\RSY.
For quasi-self-stabilization, additional compatibility
conditions are needed, as stated below.
The following lemma can be derived from the definitions of self-stabilization and quasi-self-stabilization.

\begin{lemma}\label{lem:scheduler-composition}
\begin{enumerate}
\item[(1)] A Rendezvous algorithm \(\mathcal A\) is self-stabilizing under \RSY\
if and only if it is self-stabilizing under both \FSY\ and
\RR.
The same equivalence holds for R-\RSY.
%
\item[(2-1)] If a Rendezvous algorithm ${\mathcal A}$ is \QS\ in \RR\, and it is \QS\ in \FSY\ with the property that it does not reach a configuration of differently colored robots during any execution of ${\mathcal A}$ (in \FSY),  then ${\mathcal A}$ is \QS\ in \RSY.
\item[(2-2)] If a Rendezvous algorithm ${\mathcal A}$ is \SSS\ in \RR\ (resp. \FSY), and it is \QS\ in \FSY\ (resp. \RR), then  ${\mathcal A}$ is \QS\ in \RSY\ (resp. R-\RSY).
\end{enumerate} 
\end{lemma}

In the following, we obtain \SSS\ Rendezvous algorithms in \RSY\ and R-\RSY\, and a \QS\ Rendezvous algorithm in \RSY\ by using this lemma. 
The two-color \ML-algorithm presented here is not \SSS\ in \FSY, hence a \QS\ \ML-Rendezvous algorithm for R-\RSY\ is not obtained.


In what follows, we will consider Rendezvous algorithms that operate in the Non-Rigid setting. 
The following lemma reduces the Non-Rigid correctness analysis
to configurations whose distance is at most \(\delta\).

\begin{lemma}[Distance Decreasing Lemma]
\label{lem:DDL}
Suppose that, whenever the robots occupy distinct positions,
every activated robot computes a move toward either the
midpoint or the other robot.
Rendezvous occurs or there exists
\(t'>t\) such that \(dis(t')\leq\delta\) under either of the
following conditions:
\begin{enumerate}
\item[(a)]
under \FSY, every maximal consecutive block in which both
robots move to the other's position has length at most a fixed
constant;
\item[(b)]
under \RR.
\end{enumerate}
\end{lemma}

\begin{proof} 
We only need to consider rounds with \(dis(t)>\delta\).

\noindent
(a) 
In \FSY, if we start from time $t$, the possible behaviors of $r$ and $s$ are as follows:
\begin{enumerate}
\item[(1)] Both move towards the midpoint.
\item[(2)] One moves towards the midpoint and the other moves towards the other robot.
\item[(3)] Both move towards the other robots.
\end{enumerate}
In Non-Rigid movement, for (1), if they do not rendezvous, $dis(t+1)$ is reduced by at least  $2\delta$, because both robots move at least $\delta$. 
For (2), 
we can show that the distance is reduced by at least  $\delta/2$ as follows. 
If the robots do not cross, it is reduced by at least $2\delta$. Otherwise, it is reduced by at least $dis(t)/2>\delta/2$; the minimum reduction occurs when one robot, say $r$, reaches the midpoint and $s$ reaches the location of $r$ at $t$. 


In case~(3), the distance may remain unchanged.
However, every consecutive block of case~(3) has bounded
length, whereas every intervening round of case~(1) or~(2)
either achieves Rendezvous or decreases the distance by at
least \(\delta/2\).
Hence the distance becomes at most \(\delta\) after finitely
many rounds unless Rendezvous occurs earlier.

Therefore,  
there is a time $t'$ such that $dis(t') \leq \delta$ for any possible behavior.

\noindent
(b) Under \RR, if \(dis(t)>2\delta\), every non-null move
decreases the distance by at least \(\delta\).
If
\[
    \delta<dis(t)\leq2\delta,
\]
a midpoint move reaches the midpoint and leaves distance at
most \(\delta\), whereas a move-to-other action moves at least
\(\delta\) and also leaves distance at most \(\delta\).
Thus, after finitely many activations, the distance becomes at
most \(\delta\), unless Rendezvous occurs earlier.
 Thus, 
there is a time $t'$ such that $dis(t') \leq \delta$.

\end{proof}

For algorithms satisfying the corresponding hypotheses of the
Distance Decreasing Lemma, the following properties hold.

\begin{description}
\item{(1)} If $\mathcal{A}$ is \SSS\ in \FSY\ and \RR\ for $\mathrm{dis}(0) \leq \delta$, then it is \SSS\ in \RSY\ and R-\RSY\ in the Non-Rigid setting (used in Theorems~\ref{th:SS-L-non-rigid-3-color} and~\ref{th:SS-non-L-non-rigid-2-color}).

\item[(2)] If $\mathcal{A}$ is \QS\ in \FSY\ without generating configurations with different colors and \SSS\ in \RR\ for $\mathrm{dis}(0) \leq \delta$, then it is \QS\ in \RSY\ in the Non-Rigid setting (used in Theorem~\ref{th:QS-L-non-rigid-2-color}).
\end{description}

\section{Pseudocode for \textsc{QS-L-2-color}}\label{ap:QS-L-2-color}
\Newcodeline
Let $\obs_x(t)$ denote the observed color of robot $x$ at time $t$.
\begin{algorithm}[H]
\caption{QS-\ML-2-color (one-sided visibility)}
\begin{algorithmic}[1]
\Require snapshot $SS_x(t)$ and observed color $\obs_x(t)$
\If{$\obs_x(t)=A$}
    \State set color to $B$
    \State move to  \(\mathsf{Midpoint}_x(t)\)
\ElsIf{$\obs_x(t)=B$}
    \State set color to $A$
    \State move to \(\mathsf{Other}_x(t)\)
\EndIf
\end{algorithmic}
\end{algorithm}

%
%
%

\section{Correctness of \textsc{QS-L-2-color} under \RSY}\label{ap:QS-L-2-color-correctness}

We show that QS-\ML-2-color is a \QS\ \ML-Rendezvous algorithm with 2 colors in the one-sided visibility model, hence for both $\FS$ and $\FC$, in \RSY\ even in the case of Non-Rigid.

Here, we demonstrate that QS-\ML-2-color satisfies the condition of Lemma~\ref{lem:scheduler-composition}(2-1), thereby proving it works correctly in \RSY.
Since this algorithm also satisfies the condition of Distance Decreasing Lemma (Lemma~\ref{lem:DDL}),
we can assume $dis(0) \leq \delta$ without loss of generality in the rest of this section.



\begin{lemma}\label{lem:QS-L-non-rigid-2-color-FSY}
QS-\ML-2-color is a \QS\ \ML-Rendezvous algorithm
with $2$ colors in the one-sided visibility model
under \FSY\ in the Non-Rigid setting.
Moreover, during any such execution,
no mixed-color configuration appears.
\end{lemma}

\begin{proof}
Recall that in QS-\ML-2-color,
the rule associated with the observed color $A$
is “move to the midpoint and change to $B$,”
and the rule associated with the observed color $B$
is “move to the other robot and change to $A$.”

In a same-color \FSY\ execution, a simultaneous
move-to-other round changes \((B,B)\) to \((A,A)\), so
such rounds cannot occur consecutively.
The Distance Decreasing Lemma therefore applies.

Assume that the initial configuration is same-colored,
namely either $C_{AA}(0)$ or $C_{BB}(0)$.
In the one-sided visibility model,
if both robots have the same color,
then they observe the same color,
regardless of whether the observed color is their own
(\FS) or the other robot's (\FC).
Hence both robots always apply the same rule
in every \FSY\ round.

If the initial configuration is $C_{AA}(0)$,
both robots move to the midpoint and change their colors to $B$,
so Rendezvous is achieved in one round.

If the initial configuration is $C_{BB}(0)$,
both robots move to the position of the other robot
and both change their colors to $A$;
thus they remain distinct after one round,
but the configuration becomes $C_{AA}(1)$.
Hence Rendezvous is achieved in the next round.

Finally, because in every \FSY\ round
both robots observe the same color and execute the same transition,
their colors are always updated simultaneously to the same next color.
Therefore no mixed-color configuration can appear.
\end{proof}

\begin{lemma}\label{lem:QS-L-non-rigid-2-color-RR} 
QS-\ML-2-color is a \SSS\ \ML-Rendezvous algorithm
with $2$ colors in the one-sided visibility model
under \RR\ in the Non-Rigid setting.
\end{lemma}

\begin{proof}
By the Distance Decreasing Lemma, it suffices to consider an
arbitrary configuration with \(dis(0)\leq\delta\).
Without loss of generality, let \(r\) be activated first.

In the \(\FS\) interpretation, if \(r\) has color \(B\), it
moves directly to \(s\) and achieves Rendezvous.
Suppose that \(r\) has color \(A\).
It moves to the midpoint and changes to \(B\).
At the next activation, if \(s\) has color \(B\), then \(s\)
moves directly to \(r\).
Otherwise, \(s\) has color \(A\), moves to the midpoint, and
changes to \(B\); at the following activation, \(r\), now in
color \(B\), moves directly to \(s\).
Thus Rendezvous is achieved within at most three activations.

In the \(\FC\) interpretation, if \(r\) observes \(B\), it
moves directly to \(s\).
If it observes \(A\), it moves to the midpoint and changes its
own color to \(B\); the next activated robot then observes
\(B\) and moves directly to \(r\).
Thus Rendezvous is achieved within at most two activations.

Therefore the algorithm is self-stabilizing under \RR.
\end{proof}

%
%
%
%
%
%
%
%

\setcounterref{theoremduplicate}{th:QS-L-non-rigid-2-color}
\addtocounter{theoremduplicate}{-1}

\begin{theoremduplicate}
[Reprint of Theorem~\ref{th:QS-L-non-rigid-2-color} on 
page~\pageref{th:QS-L-non-rigid-2-color}]\label{duplicateth:QS-L-non-rigid-2-color}
QS-\ML-2-color is a \QS\ \ML-Rendezvous algorithm
with $2$ colors both in the $\FS$ and the $\FC$ models
under \RSY\ in the Non-Rigid setting.
\end{theoremduplicate}

\begin{proof}
We prove the statement in the unified one-sided visibility model.
By Lemma~\ref{lem:scheduler-composition}
and Lemmas~\ref{lem:QS-L-non-rigid-2-color-FSY}
and~\ref{lem:QS-L-non-rigid-2-color-RR},
QS-\ML-2-color works in the one-sided visibility model.
Hence it works in both $\FS$ and $\FC$.
\end{proof}
%
%

\section{Pseudocode for \textsc{SS-L-3-color}}\label{ap:SS-L-3-color}
\begin{algorithm}[H]
\caption{SS-\ML-3-color (one-sided visibility)}
\begin{algorithmic}[1]
\Require snapshot $SS_x(t)$, observed color $\obs_x(t)$

\If{$\obs_x(t)=A$}
    \State set color to $B$
    \State move to  \(\mathsf{Midpoint}_x(t)\)

\ElsIf{$\obs_x(t)=B$}
    \State set color to $C$
    \State move to  \(\mathsf{Midpoint}_x(t)\)

\ElsIf{$\obs_x(t)=C$}
    \State set color to $A$
    \State move to \(\mathsf{Other}_x(t)\)
\EndIf

\end{algorithmic}
\end{algorithm}

%
%
%

\section{Detailed Analysis of
\texorpdfstring{\textsc{SS-\ML-3-color}}
{SS-L-3-color} under 
\texorpdfstring{\RSY}{RSYNCH} and \texorpdfstring{R-\RSY}{RRSYNCH}}
\label{app:SS-L-3-color-correctness}

This appendix gives the complete proof of
Theorem~\ref{th:SS-L-non-rigid-3-color} for the \(\FS\) model.

The \FSY\ argument depends only on the observed colors and
therefore applies uniformly to both \(\FS\) and \(\FC\).
For completeness, we also analyze the \(\FC\) interpretation
under \RR\ and show that the same graph solves Rendezvous for
\(\FC\) under \RSY and R-\RSY.
The latter observation is supplementary and is not needed for
the three-color upper bound, which is already known under
\SSY~\cite{FSVY}.

The algorithm SS-\ML-3-color is shown in Figure~\ref{fig:all-algorithms}(b),
and its pseudocode is given in Appendix~\ref{ap:SS-L-3-color}.
It works under both \RSY\ and R-\RSY in the case of Non-Rigid
even when starting from any initial configuration.


In the following proofs, let
$
\widehat{C}_{\alpha\beta}(t)
$
denote a configuration at time $t$
in which robot $r$ observes color $\alpha$
and robot $s$ observes color $\beta$.
In the $\FS$ interpretation, $\widehat{C}_{\alpha\beta}(t)=C_{\alpha\beta}(t)$,
whereas in the $\FC$ interpretation, $\widehat{C}_{\alpha\beta}(t)=C_{\beta\alpha}(t)$.

Since SS-\ML-3-color is defined with respect to observed colors,
the same transition graph applies to both interpretations.
\begin{lemma}\label{lem:SS-L-non-rigid-3-color-FSY}
SS-\ML-3-color is a \SSS\ \ML-Rendezvous algorithm
with $3$ colors in the one-sided visibility model
under \FSY\ in the Non-Rigid setting.
\end{lemma}
\begin{proof}
We first show that SS-\ML-3-color satisfies Distance Decreasing Lemma~(a).
The only situation in which both robots move to the current position
of the other robot is $\widehat{C}_{CC}(t)$,
and after that round the observed-color configuration becomes
$\widehat{C}_{AA}(t+1)$.
Hence this behavior never occurs twice consecutively.
Therefore Distance Decreasing Lemma~(a) applies,
and it suffices to consider the case $dis(0)\le \delta$.

Now consider any initial observed-color configuration.

If the initial configuration is one of $\widehat{C}_{AA}(0),\ \widehat{C}_{AB}(0),\ \widehat{C}_{BA}(0),\ \widehat{C}_{BB}(0)$,
then both robots move to the midpoint,
and Rendezvous is achieved in one round.

If the initial configuration is one of $\widehat{C}_{AC}(0),\ \widehat{C}_{CA}(0),\ \widehat{C}_{CC}(0)$,
then Rendezvous is achieved in at most two rounds:
from $\widehat{C}_{CC}(0)$ the next configuration is
$\widehat{C}_{AA}(1)$,
while from $\widehat{C}_{AC}(0)$ or $\widehat{C}_{CA}(0)$
the next observed-color configuration belongs to $\{\widehat{C}_{AB}(1),\widehat{C}_{BA}(1)\}$,
from which Rendezvous occurs in the following round.

Finally, if the initial configuration is
$\widehat{C}_{BC}(0)$ or $\widehat{C}_{CB}(0)$,
then after one round the resulting observed-color configuration belongs to $
\{\widehat{C}_{AC}(1),\widehat{C}_{CA}(1)\}$,
and hence Rendezvous is achieved within two more rounds.

Therefore SS-\ML-3-color attains Rendezvous
from any initial configuration in \FSY.
\end{proof}

\begin{lemma}\label{lem:SS-L-non-rigid-3-color-RR}
SS-\ML-3-color is a \SSS\ \ML-Rendezvous algorithm
with $3$ colors in the one-sided visibility model
under \RR\ in the Non-Rigid setting.
\end{lemma}
\begin{proof}
Since Distance Decreasing Lemma~(b) applies,
it suffices to consider configurations with $dis(0)\le \delta$.
Without loss of generality, assume that robot $r$ is activated first.

We distinguish the two interpretations.

\smallskip
\noindent
\emph{$\FS$ interpretation.}
Under $\FS$, each robot observes its own color.
Thus, whenever a robot is activated,
its observed color follows the progression
\[
A \rightarrow B \rightarrow C \rightarrow A,
\]
where the moves associated with $A$ and $B$ are toward the midpoint,
and the move associated with $C$ is toward the other robot.
Hence:
(i) if the activated robot already has observed color $C$,
Rendezvous is achieved in that activation;
(ii) if it has observed color $B$,
it changes to $C$ and Rendezvous is achieved at its next activation;
and (iii) if it has observed color $A$,
it changes to $B$, then to $C$, and Rendezvous is achieved
within its next two activations.
Since activations alternate in \RR,
Rendezvous is achieved from any initial configuration
within finitely many rounds.

\smallskip
\noindent
\emph{$\FC$ interpretation.}
Under $\FC$, the active robot observes the current color of the other robot.
If the active robot observes $C$,
it moves directly to the other robot,
and because $dis(0)\le \delta$,
Rendezvous is achieved in that activation.
If it observes $B$,
it changes its own color to $C$ while moving to the midpoint;
hence the other robot observes $C$ at the next activation
and Rendezvous is achieved then.
If it observes $A$,
it changes its own color to $B$ while moving to the midpoint;
then the other robot observes $B$ at the next activation
and changes its own color to $C$ while moving to the midpoint;
finally, at the following activation,
the first robot observes $C$ and reaches the other robot.
Thus Rendezvous is achieved in at most three activations.

Therefore SS-\ML-3-color attains Rendezvous
from any initial configuration in \RR.
\end{proof}

\setcounterref{theoremduplicate}{th:SS-L-non-rigid-3-color}
\addtocounter{theoremduplicate}{-1}

\begin{theoremduplicate}
[Reprint of Theorem~\ref{th:SS-L-non-rigid-3-color} on 
page~\pageref{th:SS-L-non-rigid-3-color}]\label{duplicate:th:SS-L-non-rigid-3-color}
\textsc{SS-\ML-3-color} is a three-color self-stabilizing
\ML-Rendezvous algorithm for \FS\ robots under both
\RSY\ and R-\RSY\ in the Non-Rigid setting.
\end{theoremduplicate}

\begin{proof}
Correctness under \FSY\ follows from
Lemma~\ref{lem:SS-L-non-rigid-3-color-FSY}, and correctness under
\RR\ follows from
Lemma~\ref{lem:SS-L-non-rigid-3-color-RR}.
The scheduler-composition and distance-reduction principles of
Appendix~\ref{ap:rsynch-reduction} therefore establish the
claim under both Non-Rigid \RSY and R-\RSY.
\end{proof}

\begin{proposition}[Supplementary \(\FC\) interpretation]
\label{prop:SS-L-3-color-FCOM-RSYNCH}
\textsc{SS-\ML-3-color} also solves self-stabilizing
Rendezvous for \(\FC\) robots under both \RSY\ and R-\RSY in the Non-Rigid
setting.
\end{proposition}

\begin{proof}
Correctness under \FSY\ follows from
Lemma~\ref{lem:SS-L-non-rigid-3-color-FSY}, and correctness under
\RR\ follows from
Lemma~\ref{lem:SS-L-non-rigid-3-color-RR}.
The same reduction principles therefore give correctness under
both Non-Rigid \RSY and R-\RSY.
\end{proof}

\section{Pseudocode for \textsc{SS-non-L-2-color}}\label{ap:SS-non-L-2-color}
\begin{algorithm}[H]
\caption{SS-non-L-2-color (one-sided visibility)}
\begin{algorithmic}[1]
\Require snapshot \SS$_x(t)$ and observed color $\obs_x(t)$
\If{$dis(t)=0$}
    \State \(\mathsf{Keep}\)
    \State stay
\ElsIf{$\obs_x(t)=A$}
    \If{$q_x(t)=0$} \Comment{EQI}
        \State set color to $B$
        \State move to \(\mathsf{Midpoint}_x(t)\)
    \Else \Comment{$q_x(t)=1$ (OQI)}
        \State set color to $A$
        \State move to \(\mathsf{Midpoint}_x(t)\)
    \EndIf
\ElsIf{$\obs_x(t)=B$}
    \State set color to $A$
    \State move to \(\mathsf{Other}_x(t)\)
\EndIf
\end{algorithmic}
\end{algorithm}

\section{Correctness of \textsc{SS-non-L-2-color} under \RSY\ and R-\RSY}\label{ap:Proof-SS-non-L-2-color}

The proof is essentially the same as that of \textsc{QS-\ML-2-color}.
The additional mechanism is precisely the periodic \((4,2)\)-distance class introduced in Section~\ref{subsubsec:periodic_distance_classes}.

\begin{lemma}\label{lem:SS-L-non-rigid-2-color-FSY}
Under VD+FUD, SS-non-\ML-2-color is a \SSS\ non-\ML-Rendezvous algorithm
with $2$ colors in the one-sided visibility model
under \FSY\ in the Non-Rigid setting.
\end{lemma}
\begin{proof}
In SS-non-\ML-2-color,
the only situation in which both robots move to the current position
of the other robot is $\widehat{C}_{BB}(t)$,
and after that round the observed-color configuration becomes
$\widehat{C}_{AA}(t+1)$.
Hence this behavior never occurs twice consecutively.
By Distance Decreasing Lemma~(a),
it suffices to consider the case $dis(0)\le \delta$.

We consider the possible observed-color configurations.

\smallskip
\noindent
\emph{Case 1: $\widehat{C}_{AA}(0)$.}
Both robots move to the midpoint.
Hence Rendezvous is achieved in one round,
regardless of whether the distance satisfies an even or an odd quartered interval.

\smallskip
\noindent
\emph{Case 2: $\widehat{C}_{BB}(0)$.}
Both robots move to the current position of the other robot
and both change their colors to $A$.
Thus the next observed-color configuration is $\widehat{C}_{AA}(1)$,
and Rendezvous is achieved in the next round.

\smallskip
\noindent
\emph{Case 3: the initial observed-color configuration is mixed,
that is, $\widehat{C}_{AB}(0)$ or $\widehat{C}_{BA}(0)$.}
Suppose first that $dis(0)$ satisfies an odd quartered interval for the robot observing $A$.
Then the robot observing $A$ sets its color to $A$
and moves to the midpoint,
whereas the robot observing $B$ changes to $A$
and moves to the current position of the other robot.
Hence after one \FSY\ round
the observed-color configuration becomes $\widehat{C}_{AA}(1)$,
and Rendezvous is achieved in the following round.

Now suppose that $dis(0)$ satisfies an even quartered interval  for the robot, say $r$ observing $A$.
Then in the first \FSY\ round
the robot $r$ observing $A$ changes to $B$ while moving to the midpoint,
whereas the robot observing $B$ changes to $A$
while moving to the current position of the other robot.
Hence the resulting observed-color configuration is still mixed,
and the distance is halved.
If a same-color configuration is reached earlier, the
preceding cases apply.
Otherwise, the configuration remains mixed for two rounds,
the distance becomes \(dis(0)/4\), and \(r\) again observes
\(A\).
Its periodic class has therefore changed from EQI to OQI,
so the next round produces a same-color configuration.

Thus SS-non-\ML-2-color attains Rendezvous
from any initial configuration in \FSY.
\end{proof}

\begin{lemma}
\label{lem:SS-L-non-rigid-2-color-RR}
Under VD+FUD, SS-non-\ML-2-color is a \SSS\
non-\ML-Rendezvous algorithm with two colors in the
one-sided visibility model under \RR\ in the Non-Rigid
setting.
\end{lemma}

\begin{proof}
Since part~(b) of the Distance Decreasing Lemma applies,
it suffices to consider an arbitrary configuration with
\(D=dis(0)\leq\delta\).
If \(D=0\), the zero-distance rule keeps the robots together.
Assume therefore that \(0<D\leq\delta\).
Under \RR, subsequent moves do not increase the distance,
so every move reaches its destination.
Without loss of generality, assume that robot \(r\) is
activated first.

We distinguish the two interpretations.

\smallskip
\noindent
\emph{\(\FS\) interpretation.}
If \(r\) observes color \(B\), it moves directly to the
stationary robot \(s\), and Rendezvous is achieved immediately.
Assume therefore that \(r\) observes color \(A\).

If \(D\) satisfies EQI for \(r\), then \(r\) changes its
color to \(B\) and moves to the midpoint.
Unless Rendezvous occurs earlier, \(r\) retains color \(B\)
until its next activation, when it moves to the stationary
peer and achieves Rendezvous.

Suppose instead that \(D\) satisfies OQI for \(r\).
Robot \(r\) keeps color \(A\) and moves to the midpoint.
If \(s\) has color \(B\), it reaches \(r\) at the next
activation.
Otherwise, \(s\) has color \(A\) and also moves to the
midpoint, setting its own color to either \(A\) or \(B\)
according to its distance class.

After these two midpoint moves, the physical distance is
\(D/4\), and \(r\) still has color \(A\).
By Lemma~\ref{lem:periodic-distance-class-shift},
\[
    \kappa_{r,4}^{(2)}(D/4)
    =
    1-\kappa_{r,4}^{(2)}(D)
    =
    0.
\]
Thus \(D/4\) satisfies EQI for \(r\).
At its next activation, \(r\) changes its color to \(B\)
and moves to the midpoint.
Unless Rendezvous occurs at the intervening activation of
\(s\), robot \(r\) reaches the stationary peer at its
following activation.

\smallskip
\noindent
\emph{\(\FC\) interpretation.}
If \(r\) observes color \(B\), it moves directly to the
stationary robot \(s\), and Rendezvous is achieved immediately.
Assume therefore that \(r\) observes color \(A\).

If \(D\) satisfies EQI for \(r\), then \(r\) sets its own
color to \(B\) and moves to the midpoint.
Robot \(s\) therefore observes \(B\) at the next activation
and moves directly to \(r\), achieving Rendezvous.

Suppose instead that \(D\) satisfies OQI for \(r\).
Robot \(r\) sets its own color to \(A\) and moves to the
midpoint.
At the next activation, \(s\) necessarily observes \(A\)
and also moves to the midpoint.
It sets its own color to either \(B\) or \(A\), according
to whether \(D/2\) satisfies EQI or OQI for \(s\).

After these two midpoint moves, the physical distance is
\(D/4\).
We distinguish the color observed by \(r\) at its next
activation.

If \(r\) observes \(B\), then \(s\) has set its color to
\(B\).
Robot \(r\) moves directly to the stationary robot \(s\),
and Rendezvous is achieved immediately.

Otherwise, \(r\) observes \(A\).
Since \(D\) satisfied OQI for \(r\), Lemma~\ref{lem:periodic-distance-class-shift}
gives
\[
    \kappa_{r,4}^{(2)}(D/4)
    =
    1-\kappa_{r,4}^{(2)}(D)
    =
    0.
\]
Hence \(D/4\) satisfies EQI for \(r\).
Robot \(r\) therefore sets its own color to \(B\) and moves
to the midpoint.
At the next activation, \(s\) observes \(B\) and moves
directly to \(r\), achieving Rendezvous.

Thus, in both interpretations, Rendezvous is achieved in
finitely many activations from any configuration with
distance at most \(\delta\).
The zero-distance rule then keeps the robots together.
Together with the Distance Decreasing Lemma, this proves
self-stabilization under \RR\ in the Non-Rigid setting.
\end{proof}

\section{Pseudocode for \textsc{NS-FSTA-3-color}}
\label{app:NS-FSTA-3color}

\begin{algorithm}[H]
\caption{\textsc{NS-FSTA-3-color}}
\label{alg:NS-FSTA-3color}
\begin{algorithmic}[1]
\Require snapshot \SS$_x(t)$ and own color $c_x(t)$
\If{$dis(t)=0$}
    \State \(\mathsf{Keep}\)
    \State stay
\ElsIf{$c_x(t)=A$}
    \State $i\gets\kappa_x^{(2)}(t)$
    \State set color to $B_i$
    \State move to \(\mathsf{Other}_x(t)\)
\ElsIf{$c_x(t)=B_i$ for some $i\in\{0,1\}$}
    \If{$\kappa_x^{(2)}(t)=i$}
        \State set color to $A$
        \State move to \(\mathsf{Midpoint}_x(t)\)
    \Else
         \State \(\mathsf{Keep}\)
        \State stay
    \EndIf
\EndIf
\end{algorithmic}
\end{algorithm}

\section{Correctness and Optimality of \textsc{NS-FSTA-3-color}}\label{ap:NS-FSTA-3-color-correctness}

\setcounterref{theoremduplicate}{th:NS-FSTA-SSY-3color-optimal}
\addtocounter{theoremduplicate}{-1}

\begin{theoremduplicate}
[Reprint of Theorem~\ref{th:NS-FSTA-SSY-3color-optimal} on 
page~\pageref{th:NS-FSTA-SSY-3color-optimal}]
\label{duplicate:th:NS-FSTA-SSY-3color-optimal}
Under VD+FUD and Rigid movement,
\textsc{NS-FSTA-3-color} solves Rendezvous for two $\FS$
robots under \SSY\ from the designated initial color \(A\).
Consequently, the non-quasi-self-stabilizing color complexity
of Rendezvous in this setting is exactly three.
\end{theoremduplicate}

\begin{proof}
Let the two robots be \(r\) and \(s\), and suppose first that
they occupy distinct positions.
Let their initial physical distance be \(D>0\).
Initially,
$c_r(0)=c_s(0)=A.$

Consider the first \SSY\ round.

Case 1: Exactly one robot is activated.

Without loss of generality, suppose that only \(r\) is activated.
Let
$
    i=\kappa_r^{(2)}(D).
$
Since \(r\) has color \(A\), it changes its color to \(B_i\)
and moves to the current position of \(s\).
Robot \(s\) does not move in this round.
Since movements are Rigid, \(r\) reaches \(s\), and Rendezvous
is achieved.

Case 2: Both robots are activated.

Let
$
    i=\kappa_r^{(2)}(D)
    \:\text{and}\:
    j=\kappa_s^{(2)}(D).
$
The values \(i\) and \(j\) need not be equal, since the two robots
need not agree on their unit distances.

Both robots have color \(A\), and hence they execute
$
    r:\quad A\longrightarrow B_i,\ \mathsf{Other},
$
and
$
    s:\quad A\longrightarrow B_j,\ \mathsf{Other}.
$
Thus, they exchange their positions.
The physical distance remains \(D\), and the resulting colors are
\(B_i\) and \(B_j\).

If both robots are activated in the next round, then neither
robot's distance class has changed:
$
    \kappa_r^{(2)}(D)=i,
    \:
    \kappa_s^{(2)}(D)=j.
$
Hence both robots move to the midpoint and change their colors
to \(A\).  Since the movement is Rigid, they reach the same
midpoint and achieve Rendezvous.

It remains to consider the case in which only one robot is
activated before the other one performs its next non-null move.
Without loss of generality, suppose that \(r\) is activated first.
Since
$
    \kappa_r^{(2)}(D)=i,
$
robot \(r\) changes its color from \(B_i\) to \(A\) and moves to
the midpoint.
Consequently, the physical distance becomes \(D/2\).

By Lemma~\ref{lem:periodic-distance-class-shift},
$
    \kappa_s^{(2)}(D/2)
    =
    1-\kappa_s^{(2)}(D)
    =
    1-j
    \neq j.
$
Robot \(s\) still has color \(B_j\).
Therefore, whenever \(s\) is activated while the distance remains
\(D/2\), it keeps color \(B_j\) and stays at its current position.

Robot \(r\), on the other hand, now has color \(A\).
By fairness, \(r\) is activated again after finitely many rounds.
When this happens, it changes its color to
$
    B_{\kappa_r^{(2)}(D/2)}
$
and moves to the current position of \(s\).
Since \(s\) stays at that position and the movement is Rigid,
Rendezvous is achieved.

The argument is symmetric if \(s\) is the first robot activated
after the position exchange.

Finally, if the two robots initially occupy the same position,
Rendezvous has already been achieved.
Therefore every fair \SSY\ execution starting with both robots
in color \(A\) achieves Rendezvous in finite time.

Moreover, the two-color impossibility of Theorem~\ref{thm:ssynch-known-lower-bounds}~(1) gives the
matching lower bound.
Therefore, three colors are both necessary and sufficient.
\end{proof}
%
%

%
\section{Applicability and Limitations of the \texorpdfstring{\FC}{FCOM} Interpretation
of \textsc{NS-FSTA-3-color}}\label{ap:NS-FSTA-3-color-on-FCOM}
%
The three-color transition graph also admits an \(\FC\)
interpretation, but it does not work under FUD alone.
The reason is that, in the $\FS$ interpretation, a color \(B_i\)
stores the distance class previously measured by the same robot,
whereas in the \FC interpretation, \(B_i\) communicates the
distance class measured by the other robot.
Without agreement on the unit distance, the two indices need
not coincide at the same physical distance.

Indeed, choose a physical distance \(D>0\) and fixed unit
distances such that
\[
    \kappa_r^{(2)}(D)=0,
    \qquad
    \kappa_s^{(2)}(D)=1.
\]
Start with both robots in color \(A\), and let both robots take
their initial snapshots before either of them performs its
corresponding move.
They then change their colors to \(B_0\) and \(B_1\),
respectively, and move to the initial position of the other
robot.
After both Rigid moves have been completed, the robots have
exchanged their positions, the physical distance is still \(D\),
and their colors are \(B_0\) and \(B_1\).

Robot \(r\) observes \(B_1\), but its current distance class is
\(0\), and hence it stays while keeping its current color.
Symmetrically, robot \(s\) observes \(B_0\), but its current
distance class is \(1\), and hence it also stays while keeping
its current color.
Therefore,
\[
    (B_0,B_1;D)
    \longrightarrow
    (B_0,B_1;D)
    \longrightarrow
    \cdots
\]
is a fair non-Rendezvous execution.
Thus the direct \(\FC\) interpretation is not correct under
FUD without agreement on the unit distance, even under
CM-atomic \ASY\ and Rigid movement.

If the robots agree on their unit distance, however, then
$
    \kappa_r^{(2)}(D)
    =
    \kappa_s^{(2)}(D)
$
for every physical distance \(D>0\).
The phase mismatch exhibited above can therefore never occur,
and the same three-color graph solves Rendezvous for
\(\FC\) robots under CM-atomic \ASY\ and Rigid movement.

After the initial swap, both visible indices still correspond to 
the unchanged physical distance.
If both robots react before the distance changes, they move to
the same midpoint.
If only one robot completes its midpoint move first, the
distance class is reversed: the midpoint robot stays, while
the peer observes color \(A\), publishes the new class, and
moves to the midpoint robot.
$CM$-atomicity prevents an incompatible intervening observation.

Notice that the common-unit assumption does not by itself make
the same protocol correct under $M$-atomic \ASY.
The $CM$-atomic assumption is still used to prevent a robot from
observing a newly published phase color together with the
position preceding the associated move.
A six-color non-quasi-self-stabilizing Rendezvous algorithm
for $\FS$ robots under \SSY, Rigid movement, Fixed
Disorientation, and FUD is known from~\cite{FSVY}.
Our algorithm strengthens this result in two respects.
First, it removes the Fixed Disorientation assumption and works
under VD+FUD, where the direction and orientation of the local
coordinate system may change arbitrarily from one \Look to the next.
Second, it reduces the number of colors from six to three.
The key idea is to replace the directional information used in
the previous algorithm by a one-bit periodic encoding of the
distance.  Since the unit distance of each robot is fixed under
FUD, halving the physical distance reverses this bit independently
of the current orientation of the local coordinate system.
Together with the matching two-color impossibility result, this
shows that three colors are optimal under VD+FUD.

\section{Pseudocode for
\texorpdfstring{\textsc{NS-FSTA-4-color}}
{NS-FSTA-4-color}}
\label{app:NS-FSTA-4-color}

In the \(\FS\) model, robot \(x\) observes its own color;
that is, \(\obs_x(t)=c_x(t)\).
All phase indices are taken modulo \(3\).

\begin{algorithm}[H]
\caption{\textsc{NS-FSTA-4-color}}
\label{alg:NS-FSTA-4-color}
\begin{algorithmic}[1]
\Require snapshot \SS\(_x(t)\) and own color \(c_x(t)\)

\If{\(dis(t)=0\)}
    \State \(\mathsf{Keep}\)
    \State stay
\Else
    \State \(i\gets\kappa_x^{(3)}(t)\)

    \If{\(c_x(t)=I\)}
        \State \(c_x\gets C_i\)
        \State move to \(\mathsf{Other}_x(t)\)

    \ElsIf{\(c_x(t)=C_j\) for some \(j\in\mathbb Z_3\)}
        \State \(\Delta\gets(i-j)\bmod 3\)

        \If{\(\Delta=0\)}
            \State \(c_x\gets C_{(j+1)\bmod 3}\)
            \State move to \(\mathsf{Midpoint}_x(t)\)

        \ElsIf{\(\Delta=1\)}
            \State \(\mathsf{Keep}\)
            \State stay

        \Else
            \Comment{\(\Delta=2\)}
            \State \(\mathsf{Keep}\)
            \State move to \(\mathsf{Other}_x(t)\)
        \EndIf
    \EndIf
\EndIf
\end{algorithmic}
\end{algorithm}

\section{Pseudocode for
\texorpdfstring{\textsc{NS-FCOM-4-color}}
{NS-FCOM-4-color}}
\label{app:NS-FCOM-4-color}

In the \(\FC\) model, robot \(x\) observes the color displayed
by the other robot; that is,
\(\obs_x(t)=c_{\bar{x}}(t)\).
Every assignment to \(c_x\) below updates the light of \(x\),
not the observed light of the peer.
All phase indices are taken modulo \(3\).

\begin{algorithm}[H]
\caption{\textsc{NS-FCOM-4-color}}
\label{alg:NS-FCOM-4-color}
\begin{algorithmic}[1]
\Require snapshot \SS\(_x(t)\) and observed color \(\obs_x(t)\)

\If{\(dis(t)=0\)}
    \State \(\mathsf{Keep}\)
    \State stay
\Else
    \State \(i\gets\kappa_x^{(3)}(t)\)

    \If{\(\obs_x(t)=I\)}
        \State \(c_x\gets C_i\)
        \State move to \(\mathsf{Other}_x(t)\)

    \ElsIf{\(\obs_x(t)=C_j\) for some \(j\in\mathbb Z_3\)}
        \State \(\Delta\gets(i-j)\bmod 3\)

        \If{\(\Delta=0\)}
            \State \(c_x\gets C_{(j+1)\bmod 3}\)
            \State move to \(\mathsf{Midpoint}_x(t)\)

        \ElsIf{\(\Delta=1\)}
            \State \(\mathsf{Keep}\)
            \State move to \(\mathsf{Other}_x(t)\)

        \Else
            \Comment{\(\Delta=2\)}
            \State \(\mathsf{Keep}\)
            \State stay
        \EndIf
    \EndIf
\EndIf
\end{algorithmic}
\end{algorithm}

\section{Correctness of the Four-Color Phase Algorithms under Atomic \ASY}
\label{aP:subsec:four-color-atomic}

We first present a four-color non-quasi-self-stabilizing
Rendezvous algorithm for \(\FS\) robots under $M$-atomic
\ASY.
We then adapt the same periodic-phase mechanism to the
\(\FC\) model under $CM$-atomic \ASY.

The two algorithms use the same color set, the same periodic
three-way distance classes, and the same swap initialization.
They do not, however, use exactly the same phase-to-action
mapping.
The two nonzero phase offsets play opposite roles because a
phase color represents private temporal information in the
\(\FS\) interpretation, but information communicated by the
peer in the \(\FC\) interpretation.

The \(\FS\) algorithm, denoted by \textsc{NS-FSTA-4-color}, uses
$    \{I,C_0,C_1,C_2\},$
where \(I\) is the designated initial color.
Since the robots are in the \(\FS\) model, each robot observes
its own color.

The algorithm uses the periodic three-way distance classes
$
    \kappa_x^{(3)}(t)
    =
    \kappa_{x,2}^{(3)}(t)
    \in\mathbb Z_3
$
introduced in Section~\ref{subsec:algorithmic-notation}.
Recall that
\begin{equation}
\label{eq:four-color-halving}
    \kappa_x^{(3)}(D/2)
    =
    \kappa_x^{(3)}(D)-1
    \pmod 3.
\end{equation}

If robot \(x\) has color \(I\), it sets
$
    i=\kappa_x^{(3)}(t),
$
changes its color to \(C_i\), and moves to the current position
of the other robot.

Suppose next that \(x\) has color \(C_j\), and define
$
    \Delta_x(t,j)
    =
    \kappa_x^{(3)}(t)-j
    \pmod 3.
$
Robot \(x\) acts according to
\[
\begin{array}{c|c|c}
\Delta_x(t,j)
    &\text{light update}&\text{destination}\\ \hline
0&C_{j+1}&\mathsf{Midpoint}_x(t)\\
1&\mathsf{Keep}&\mathsf{Stay}_x(t)\\
2&\mathsf{Keep}&\mathsf{Other}_x(t).
\end{array}
\tag{\(\mathcal A^{\FS}_4\)}
\label{tab:four-color-FSTA-rules}
\]
All indices are taken modulo \(3\).

The three cases have a temporal interpretation.
Color \(C_j\) records a distance class previously measured by
the robot itself.
If the current class still agrees with the recorded index, the
robot moves to the midpoint and advances the recorded phase.
After a one-sided midpoint move, the mover obtains offset \(1\)
and waits, whereas the robot that has not moved obtains offset
\(2\) and moves to the waiter.

The graph representation of \textsc{NS-FSTA-4-color} is shown
in Figure~\ref{fig:new-algorithm-graphs}~(b), and its
pseudocode is given in
Appendix~\ref{app:NS-FSTA-4-color}.

In the following correctness proofs, we have to consider the transient coincidence during the initial swap.
Suppose that both robots have taken their \(I\)-snapshots at
physical distance \(D\), so that both have incomplete cycles
based on the same initial geometry.
Each such cycle computes the peer's old position as its
destination.

If one robot completes its swap move first, the robots may
temporarily occupy the same position.
This need not yet be terminal Rendezvous, since the other robot
may subsequently execute its pending move to the first robot's
old position.

Any \Look performed while the physical distance is zero is
harmless: the zero-distance rule has priority, and hence the
activated robot neither changes its color nor moves.
Thus, the robot that has arrived first may complete arbitrarily
many null cycles without affecting the pending swap.

In the \(\FS\) proof, M-atomicity prevents any \Look during the
peer's pending \Move.
In the \(\FC\) proof, CM-atomicity prevents any \Look from the
peer's \Compute operation until the completion of its pending
\Move.
Consequently, once the second pending Rigid move is completed,
the robots have exchanged positions, their physical distance
is again \(D\), and their phase colors are exactly those
computed from the original \(I\)-snapshots.

\begin{lemma}
\label{lem:four-color-FSTA}
Under VD+FUD and Rigid movement,
\textsc{NS-FSTA-4-color} solves Rendezvous for two $\FS$ robots
under M-atomic \ASY\ from the designated initial color \(I\).
\end{lemma}

\begin{proof}
Let the robots be \(r\) and \(s\), initially at distinct
positions at physical distance \(D>0\), with
$
    c_r=c_s=I.
$

Since the model is $\FS$, each robot observes its own color.
Moreover, under $M$-atomic \ASY, no \Look can occur while the
other robot is moving.
Notice also that a color change of an $\FS$ robot is invisible
to the other robot.  Hence the interval between \Compute and
\Move-Begin creates no additional observation that must be
considered.

Consider first an activation in color \(I\).
If only one robot has taken its \(I\)-snapshot before the first
initial \Move begins, \(M\)-atomicity prevents the peer from
performing a \Look during that \Move.
Hence the mover reaches the stationary peer and Rendezvous is
achieved.

Otherwise, both robots have taken their \(I\)-snapshots before
either initial \Move begins.


Let
$
    i=\kappa_r^{(3)}(D),
    \:
    j=\kappa_s^{(3)}(D).
$
The two robots set their colors to \(C_i\) and \(C_j\),
respectively, and each has an incomplete swap cycle whose
destination is the old position of the other robot.
By the preceding observation, a temporary coincidence caused
by the first completed swap move is harmless, even if the
robot that arrives first is activated repeatedly at distance
zero.

After both pending swap moves have been completed, the robots
have exchanged positions and the physical distance is again
\(D\).

Thus,
\[
    \kappa_r^{(3)}(D)=i,
    \qquad
    \kappa_s^{(3)}(D)=j.
\]


Consequently, every action computed from a fresh snapshot at
physical distance \(D\) is a midpoint action.
Let \(r\) be the first robot to complete such a midpoint move,
and let \(m\) be the corresponding physical midpoint.
If the peer has already begun its pending midpoint move,
\(M\)-atomicity prevents any fresh Look of \(r\) until that
move is completed, at which point Rendezvous is achieved.
Otherwise, the peer is still at its old position and the
physical distance is \(D/2\).

We distinguish whether \(s\) has already performed a \Look of
the phase-consistent configuration at distance \(D\) before
\(r\) begins its midpoint move.

\noindent
{\bf Case 1: \(s\) has already performed such a \Look.}

The incomplete cycle of \(s\), if not yet completed, is based
on the same phase-consistent configuration at distance \(D\).
Hence \(s\) computes a midpoint move to the same physical point
\(m\) and changes its color from \(C_j\) to \(C_{j+1}\).

After \(r\) reaches \(m\), it may be activated arbitrarily many
times before \(s\) completes its pending midpoint move.
At this time,
$
    \kappa_r^{(3)}(D/2)=i-1,
$
whereas \(r\) has color \(C_{i+1}\).  Thus,
$
    (i-1)-(i+1)
    \equiv 1 \pmod 3.
$
Therefore, every such activation makes \(r\) keep color
\(C_{i+1}\) and stay at \(m\).

Once \(s\) begins its pending \Move, \(M\)-atomicity prevents
\(r\) from performing a \Look while \(s\) is moving.
Since the movement of \(s\) is Rigid, \(s\) eventually reaches
the same midpoint \(m\), and Rendezvous is achieved.

\noindent
{\bf Case 2: \(s\) has not performed such a \Look.}

By \(M\)-atomicity, \(s\) cannot perform a $Look$ while \(r\)
is executing its midpoint move.  Hence the next Look of \(s\)
is taken after \(r\) has reached \(m\), when the physical
distance is \(D/2\).

Before this $Look$ occurs, \(r\) may again be activated
arbitrarily many times.  However, as shown above, its phase
difference is always
$
    (i-1)-(i+1)
    \equiv1\pmod3;
$
hence every such activation is null: \(r\) keeps color
\(C_{i+1}\) and stays at \(m\).
Thus the configuration remains unchanged until \(s\) performs
its next \Look.

By fairness, \(s\) eventually performs this $Look$.
Since \(s\) still has color \(C_j\), and
$
    \kappa_s^{(3)}(D/2)=j-1,
$
its phase difference is
$
    (j-1)-j
    \equiv2\pmod3.
$
Therefore \(s\) keeps color \(C_j\) and computes a move to the
current position \(m\) of \(r\).

Even if the beginning of this \Move is delayed, every intervening
activation of \(r\) is again null, so \(r\) remains at \(m\).
Once \(s\) begins moving, \(M\)-atomicity prevents \(r\) from
performing a \Look during the \Move.
By rigidity, \(s\) reaches \(m\), and Rendezvous is achieved.

The argument is symmetric if \(s\) is the first robot to
complete a midpoint move.
Thus every fair \(M\)-atomic \ASY\ execution from the
designated initial color \(I\) attains Rendezvous in finite time.
\end{proof}



\noindent
{\bf Adaptation to the \(\FC\) model}

We now adapt the periodic-phase mechanism of
\textsc{NS-FSTA-4-color} to the \(\FC\) model.
The color set and the swap initialization remain unchanged,
but the actions associated with the two nonzero offsets are
reversed.

The reason is that a phase color has a different meaning in the
two models.
In the \(\FS\) interpretation, \(C_j\) is read by the robot
that previously stored it.
In the \(\FC\) interpretation, \(C_j\) is read by the peer and
therefore communicates the distance class measured by the robot
displaying it.

This difference becomes visible after a one-sided midpoint
move.
In the \(\FS\) algorithm, the midpoint mover compares the new
distance class with its own advanced phase color and obtains
offset \(1\).
In the \(\FC\) algorithm, the midpoint mover observes the
peer's old phase color and obtains offset \(2\).
Conversely, the peer observes the mover's advanced phase color
and obtains offset \(1\).
Thus, to make the midpoint mover stay and the peer approach it,
the \(\FC\) algorithm maps offsets \(1\) and \(2\) to
\(\mathsf{Other}\) and \(\mathsf{Stay}\), respectively.

We denote the resulting algorithm by
\textsc{NS-FCOM-4-color}.
When robot \(x\) observes \(I\), it sets
$
    i=\kappa_x^{(3)}(t),
$
changes its color to \(C_i\), and moves to the current position
of the other robot.

When \(x\) observes \(C_j\), it computes
$
    \Delta_x(t,j)
    =
    \kappa_x^{(3)}(t)-j
    \pmod 3
$
and applies the following rule:
\[
\begin{array}{c|c|c}
\Delta_x(t,j)
    &\text{light update}&\text{destination}\\ \hline
0&C_{j+1}&\mathsf{Midpoint}_x(t)\\
1&\mathsf{Keep}&\mathsf{Other}_x(t)\\
2&\mathsf{Keep}&\mathsf{Stay}_x(t).
\end{array}
\tag{\(\mathcal A^{\FC}_4\)}
\label{tab:four-color-FCOM-rules}
\]

Thus, the two four-color algorithms use the phase-to-action
mappings
\[
\begin{array}{c|ccc}
 &0&1&2\\ \hline
\textsc{NS-FSTA-4-color}&H&S&O\\
\textsc{NS-FCOM-4-color}&H&O&S.
\end{array}
\]
Although the two nonzero actions are reversed, phase
complementarity implies that both mappings generate the same
set of geometric action pairs,
\[
    H/H,\qquad O/S,\qquad S/O,
\]
up to symmetry.

By Lemma~\ref{lem:phase-complementarity}, whenever two visible
phase colors are consistent with the same physical distance,
the ordered pair of offsets is one of
$
    (0,0),\:(1,2),\:(2,1).
$
Applying the phase-to-action mapping of
\textsc{NS-FCOM-4-color} yields
$
    H/H,\: O/S,\: S/O,
$
respectively.
This is the geometric core of the following correctness result.

\begin{lemma}
\label{lem:four-color-FCOM-CM}
Under VD+FUD and Rigid movement,
\textsc{NS-FCOM-4-color} solves Rendezvous for two \(\FC\)
robots under CM-atomic \ASY\ from the designated initial
color \(I\).
\end{lemma}

\begin{proof}
Let the robots initially occupy distinct positions at physical
distance \(D>0\), and let both robots initially have color \(I\).

Suppose first that one robot completes its rule triggered by
\(I\) before the other robot takes an \(I\)-snapshot.
The moving robot reaches the stationary peer, and Rendezvous is
achieved.

Otherwise, both robots have taken their \(I\)-snapshots before
the first corresponding \Compute operation.
Let
$
    i=\kappa_r^{(3)}(D),
    \:
    j=\kappa_s^{(3)}(D).
$
Both have incomplete cycles based on their initial
\(I\)-snapshots.
These cycles set their colors to \(C_i\) and \(C_j\),
respectively, and compute move-to-other actions.
By the preceding observation, if one swap move is completed
first, any temporary coincidence and any permitted zero-distance
cycles of the robot that arrives first are harmless.
After both pending swap moves have been completed, the physical
distance is again \(D\), and the visible phase colors are
consistent with that distance:
$
    i=\kappa_r^{(3)}(D),
    \:
    j=\kappa_s^{(3)}(D).
$

By Lemma~\ref{lem:phase-complementarity}, the possible ordered
pairs of phase offsets are
$
    (0,0),\: (1,2),\: (2,1).
$
Since \textsc{NS-FCOM-4-color} maps offsets \(0,1,2\) to
\(H,O,S\), respectively, the corresponding geometric action
pairs are
$
    H/H,\: O/S,\: S/O.
$

Suppose first that the phase difference is nonzero.
Exactly one robot moves to the current position of the other
robot, while the other robot stays.
Repeated activations of the stationary robot do not change the
configuration.
By fairness and rigidity, the moving robot eventually reaches
the stationary one, and Rendezvous is achieved.

It remains to consider the case of equal phase indices.
Both robots then compute midpoint moves whenever their snapshots
are taken from the same equal-phase configuration.
Let \(r\) be the first robot to complete such a midpoint move,
and let \(m\) be the corresponding physical midpoint.
If \(s\) has already begun Compute for its pending midpoint
move, \(CM\)-atomicity prevents any fresh Look of \(r\)
until \(s\) reaches \(m\).
Otherwise, \(s\) is still at its old position and displays
\(C_i\), so the physical distance is \(D/2\).
And 
\(r\) has
color \(C_{i+1}\), and the peer still has color \(C_i\).
We distinguish whether \(s\) has an incomplete cycle whose \Look
was performed in the original equal-phase configuration.
If such a cycle exists, its destination is the same physical
midpoint \(m\).
If \(s\) has not yet performed its \Compute operation, every
fresh \Look of \(r\) gives
$
    \kappa_r^{(3)}(D/2)-i
    =
    (i-1)-i
    =
    2
    \pmod3,
$
and hence \(r\) stays at \(m\).
Once \(s\) performs \Compute, CM-atomicity prevents \(r\) from
performing another \Look until \(s\) completes its move.
Since the move of \(s\) is Rigid and has destination \(m\),
Rendezvous is achieved.

Otherwise, \(s\) takes its next \Look after \(r\) has reached
\(m\).
At distance \(D/2\), the two fresh offsets are
\(
\begin{aligned}
    \Delta_r
        &=(i-1)-i
          =2
          \pmod3,\\
    \Delta_s
        &=(i-1)-(i+1)
          =1
          \pmod3.
\end{aligned}
\)
Thus \(r\) stays at \(m\), whereas \(s\) moves to \(m\).
Any repeated activation of \(r\) before this move is completed
is again null.
By fairness and rigidity, \(s\) eventually reaches \(r\).

The case in which \(s\) is the first robot to complete the
midpoint move is symmetric.
Hence every fair $CM$-atomic \ASY\ execution from the
designated initial color \(I\) achieves Rendezvous in finite
time.
The zero-distance rule then makes the robots remain at their
common position.
\end{proof}

\section{Pseudocode for \textsc{NS-FCOM-6-color}}
\label{app:NS-FCOM-6-color}

\begin{algorithm}[H]
\caption{\textsc{NS-FCOM-6-color} for \(\FC\) under \ASY}
\label{alg:NS-6-color}
\begin{algorithmic}[1]
\Require snapshot \SS\(_x(t)\) and observed color \(\obs_x(t)\)

\If{\(dis(t)=0\)}
    \State \(\mathsf{Keep}\)
    \State stay
\Else
    \State \(i\gets\kappa_x^{(3)}(t)\)

    \If{\(\obs_x(t)=I\)}
        \State \(c_x\gets C_i\)
        \State stay

    \ElsIf{\(\obs_x(t)=C_j\) for some \(j\in\mathbb Z_3\)}
        \State \(\Delta\gets(i-j)\bmod 3\)

        \If{\(\Delta=0\)}
            \State \(c_x\gets M\) \Comment{guard a midpoint move}
            \State move to \(\mathsf{Midpoint}_x(t)\)

        \ElsIf{\(\Delta=1\)}
            \State \(c_x\gets M\) \Comment{guard a move-to-other}
            \State move to \(\mathsf{Other}_x(t)\)

        \Else
            \Comment{\(\Delta=2\)}
            \State \(c_x\gets N\) \Comment{become an anchor}
            \State stay
        \EndIf

    \ElsIf{\(\obs_x(t)=M\)}
       \State \(\mathsf{Keep}\) 
        \State stay

    \ElsIf{\(\obs_x(t)=N\)}
        \State \(c_x\gets M\)
        \State move to \(\mathsf{Other}_x(t)\)
    \EndIf
\EndIf
\end{algorithmic}
\end{algorithm}


\section{Correctness of \textsc{NS-FCOM-6-color} under Full \ASY}\label{app:NS-FCOM-6-color-correctness}
\setcounterref{lemmaduplicate}{lem:six-color-moving-guard}
\addtocounter{lemmaduplicate}{-1}

\begin{lemmaduplicate}
[Moving guard;
reprint of Lemma~\ref{lem:six-color-moving-guard} on 
page~\pageref{lem:six-color-moving-guard}]
\label{duplicate:lem:six-color-moving-guard}
Suppose that a robot \(r\) computes a non-null move in
\textsc{NS-FCOM-6-color}.
At the end of this \Compute operation, \(r\) displays color \(M\).

Until \(r\) completes the move, every fresh \Look performed by
the other robot \(s\) observes \(M\) and hence generates a
null move.
Consequently, while \(M\) remains visible, no fresh snapshot of
\(s\) can generate a new non-null move.
Robot \(s\) may have at most one incomplete cycle whose \Look
occurred before \(M\) became visible.
\end{lemmaduplicate}

\begin{proof}
Every rule producing a non-null destination sets the moving
robot's color to \(M\) during \Compute.
By the rule for observed color \(M\), every subsequent fresh
\Look of the peer generates a null move.

Since a robot executes its \LCM cycles sequentially, it can have
at most one incomplete cycle whose \Look occurred before \(M\)
became visible.
\end{proof}

\begin{lemma}
\label{ap:lem:six-color-stale-compatible}
Assume that the execution starts from \((I,I)\) at positive
physical distance.
Then a first occurrence of color \(M\) arises after finitely
many activations.

Let \(r\) be the first robot that changes its color to \(M\)
and computes a non-null move.
If the other robot \(s\) has an incomplete cycle whose \Look
occurred before this first occurrence of \(M\), then the action
determined by that snapshot is geometrically compatible with
the move of \(r\).

More precisely, treating the absence of such an action as a
null action, the possible geometric action pairs are, up to
symmetry,
$
    O/S,\: H/H,\: H/S.
$
\end{lemma}
\begin{proof}
Before the first occurrence of \(M\), no non-null move has
changed the physical distance: the rules triggered by \(I\)
and the offset-\(2\) phase rule are both null.
Consequently, every phase color \(C_i\) generated before the
first \(M\) is consistent with the same unchanged physical
distance.

We first observe that \(M\) must occur after finitely many
activations.
Eventually a robot publishes a phase color.
If the peer reacts with offset \(0\) or \(1\), it enters \(M\).
If it reacts with offset \(2\), it enters \(N\), after which
the other robot eventually observes \(N\) and enters \(M\).
If both robots first publish phase colors, phase
complementarity
(Lemma~\ref{lem:phase-complementarity}) again implies that at
least one of them has offset \(0\) or \(1\).

Now consider the first move entering \(M\).

If \(r\) observes \(N\), it computes \(O\).
Let \(C_h\) be the phase color that \(s\) observed with offset
\(2\) when it entered \(N\).
Since the distance is unchanged, every subsequent
\(I\)-snapshot of \(r\) republishes the same color \(C_h\),
and before the first \(M\), no rule can make \(r\) publish a
different phase color.
Moreover, both robots cannot enter \(N\) from phase snapshots,
because this would require both complementary offsets to be
\(2\).
Hence every pre-\(M\) snapshot of \(s\) is either an
\(I\)-snapshot or another offset-\(2\) observation of \(C_h\);
in either case, the resulting action is null.
Thus the action pair is \(O/S\).

Suppose that \(r\) observes a phase color.
If the pre-\(M\) snapshot of \(s\) observes \(I\), the action
of \(s\) is the null initialization action, yielding \(H/S\)
or \(O/S\).

Otherwise, both relevant snapshots involve phase colors
consistent with the unchanged distance.
By Lemma~\ref{lem:phase-complementarity}, the two offsets are
\(
    (0,0),\:(1,2),\:(2,1).
\)
Since the six-color algorithm maps offsets \(0,1,2\) to
\(H,O,S\), respectively, the resulting action pairs are
\(H/H\), \(O/S\), or \(S/O\).

Thus, up to symmetry, the only possible pairs are
\(
    O/S,\: H/H,\: H/S.
\)
\end{proof}

\begin{lemma}
\label{ap:lem:six-color-anchor}
Suppose that the first guarded move is a midpoint move by
robot \(r\), while the compatible pre-\(M\) action of the other
robot \(s\) is null.
Then \textsc{NS-FCOM-6-color} reaches Rendezvous in finite time.
\end{lemma}

\begin{proof}
By the compatibility argument, the \(H/S\) case arises only
when \(s\) has no pre-\(M\) action or when its pre-\(M\)
action is the null initialization rule.
In either case, \(s\) remains at its current position and
continues to display the phase color \(C_j\) observed by \(r\).
Let \(D\)
be the physical distance before the midpoint move.
After \(r\) reaches the midpoint, the distance is \(D/2\).

Since \(r\) selected a midpoint move, its distance class before
the move was \(j\).
By Lemma~\ref{lem:periodic-distance-class-shift}, its new class
is \(j-1\), and hence its offset from the still-visible color
\(C_j\) is
$
    (j-1)-j=2\pmod3.
$
Thus every activation of \(r\) before the peer changes to \(M\)
makes \(r\) display \(N\) and stay at the midpoint.

While \(r\) displays \(M\), every fresh cycle of \(s\) is
null.
Once \(r\) has reached the midpoint and changed to \(N\),
any earlier null cycle of \(s\) finishes in finite time.
Fairness then ensures that \(s\) observes \(N\), writes
\(M\), and moves to \(r\).
%

Thereafter every fresh \Look of \(r\) observes \(M\), so \(r\)
remains stationary.
By rigidity, \(s\) reaches \(r\), and Rendezvous is achieved.
\end{proof}

\setcounterref{theoremduplicate}{th:six-color-FCOM-ASYNCH}
\addtocounter{theoremduplicate}{-1}

\begin{theoremduplicate}
[Reprint of Theorem~\ref{th:six-color-FCOM-ASYNCH} on 
page~\pageref{th:six-color-FCOM-ASYNCH}]
\label{duplicate:th:six-color-FCOM-ASYNCH}
Under VD+FUD and Rigid movement,
\textsc{NS-FCOM-6-color} solves Rendezvous for two \(\FC\) robots
under \ASY from the designated initial color \(I\).
Consequently, there exists a six-color
non-quasi-self-stabilizing Rendezvous algorithm in this setting.
\end{theoremduplicate}

\begin{proof}
Consider an arbitrary fair execution starting from \((I,I)\) at
positive physical distance.

By Lemma~\ref{ap:lem:six-color-stale-compatible}, a first robot
\(r\) enters \(M\) after finitely many activations and computes
a non-null move.
By Lemma~\ref{duplicate:lem:six-color-moving-guard},
until \(r\) completes this move, every action of \(s\)
generated from a fresh snapshot is null;
at most one cycle whose \Look occurred before the first \(M\)
may remain incomplete.

By Lemma~\ref{ap:lem:six-color-stale-compatible}, the possible
geometric action pairs are, up to symmetry,
$
    O/S,\: H/H,\: H/S.
$

In the \(O/S\) case, the moving robot rigidly reaches the
stationary one.

In the \(H/H\) case, both actions have the same physical
midpoint as destination, even if one of the \Compute operations
is delayed.
After the first robot reaches that midpoint, every one of its
subsequent activations is null: it either observes \(M\), or
observes the still-visible phase color at the halved distance
and enters \(N\).
Thus it remains at the midpoint until the other Rigid midpoint
move arrives.

In the \(H/S\) case,
Lemma~\ref{ap:lem:six-color-anchor} guarantees that the midpoint
robot becomes an \(N\)-anchor and the other robot eventually
reaches it as an \(M\)-chaser.

Hence Rendezvous is achieved in finite time in every case.
The zero-distance rule then makes the robots remain at the
common position.
\end{proof}
\DE



\end{document}